\documentclass[12pt]{article}
\usepackage{fullpage}
\usepackage{amsmath}
\usepackage{amsthm}
\usepackage{amssymb}
\usepackage{abstract}
\usepackage{tensor}
\usepackage[hyperfootnotes=true,colorlinks=true,linkcolor=blue,citecolor=blue,urlcolor  = blue,anchorcolor = blue,linktocpage]{hyperref}
\numberwithin{equation}{section}
\usepackage[width=\textwidth]{caption}
\usepackage{cite} 
\usepackage[utf8]{inputenc}
\usepackage{tikz}
\usepackage[percent]{overpic}
\usepackage{float}
\usepackage{soul}
\usepackage{authblk}
\restylefloat{table}

\definecolor{FLRWcolor}{rgb}{0.2, 0, 0.2}
\definecolor{coordcolor}{rgb}{0.2, 0.3, 0.6}
\definecolor{boxcolor}{rgb}{0.3, 0.4, 0.3}
\definecolor{sfcolor}{rgb}{0.6, 0, 0.2}
\definecolor{sfacolor}{rgb}{0,0,0}
\definecolor{gcolorUn}{rgb}{1, 0.3, 0}
\definecolor{gcolorSt}{rgb}{0.3, 0.6, 0.5}

\newcommand{\be}{\begin{equation}}
\newcommand{\ee}{\end{equation}}

\newcommand{\keff}{k_\text{eff}}

\newtheorem{theorem}{Theorem} 

\newtheorem{definition}{Definition}
\newtheorem{lemma}{Lemma}

\title{Extended FLRW models and their stability}

\title{Extended FLRW Models:\\
dynamical cancellation of cosmological anisotropies}

\author{Mikjel Thorsrud$^{a}$\footnote{{\bf e-mail}: mikjel.thorsrud@hiof.no} \,, 
	Ben D. Normann$^{b}$\footnote{{\bf e-mail}: ben.d.normann@uis.no} \,,
	Thiago S. Pereira$^{c}$\footnote{{\bf e-mail}: tspereira@uel.br} 
	\\
	$^a$ \small{\em Faculty of Engineering, \O stfold University College,}\\
		\small{\em P.O. Box 700, 1757 Halden, Norway.} \\
	$^b$ \small{\em Faculty of Science and Technology, University of Stavanger,} \\
	\small{\em 4036, Stavanger, Norway.}\\
	$^c$ \small{\em Departamento de Física, Universidade Estadual de Londrina,}\\
	\small{\em Rod. Celso Garcia Cid, Km 380, 86057-970, Londrina, Paraná, Brazil.}\\}

\date{}
\begin{document}

\maketitle

\begin{abstract}
We investigate a corner of the Bianchi models that has not received much attention: ``extended FLRW models'' (eFLRW) defined as a cosmological model with underlying anisotropic Bianchi geometry that nevertheless expands isotropically and can be mapped onto a reference FLRW model with the same expansion history. In order to investigate the stability and naturalness of such models in a dynamical systems context, we consider spatially homogeneous models that contain a massless scalar field $\varphi$ and a non-tilted perfect fluid obeying an equation of state $p=w\rho$. Remarkably, we find that matter anisotropies and geometrical anisotropies tend to cancel out dynamically. Hence, the expansion is asymptotically isotropic under rather general conditions. Although extended FLRW models require a special matter sector with anisotropies that are ``fine-tuned'' relative to geometrical anisotropies, our analysis shows that such solutions are dynamically preferred attractors in general relativity. Specifically, we prove that all locally rotationally symmetric Bianchi type III universes with space-like $\nabla_\mu\varphi$ are asymptotically shear-free, for all $w\in[-1,1]$. Moreover, all shear-free equilibrium sets with anisotropic spatial curvature are proved to be stable with respect to all homogeneous perturbations for $w\geq -1/3$.
\end{abstract}


\newpage
\setcounter{tocdepth}{2}
\tableofcontents
\newpage

\section{Introduction}

The physical aspects of our Universe at the largest cosmological scales were measured with unprecedented accuracy by cosmic microwave background (CMB) experiments such as COBE~\cite{Smoot:1992td}, WMAP~\cite{Komatsu:2008hk, Komatsu:2010fb, Hinshaw:2012aka} and Planck \cite{Ade:2013zuv,Ade:2015xua,planck18}. Despite confirmations by the Planck Collaboration of CMB spectra still in tension with the predicted lensing amplitude \cite{planck18} and the presence of several large-angle statistical ``anomalies'' \cite{PlanckIsotropy:2019},\footnote{The origin of these features is still unclear and there is an ongoing discussion if they are indicating new physics or if they are merely statistical fluctuations. See \cite{DiValentino:2019} for a discussion on possible implications of the tension related to the lensing amplitude, with CMB spectra favoring a positive spatial curvature at more than the 99$\%$ confidence level. For reviews on the socalled ``$\Lambda$CDM anomalies'', see~\cite{Copi:2010na,Schwarz:2015cma,Bull:2015stt} and references therein.} the predictions of the concordance ($\Lambda$CDM) cosmological model seem to match the results of these experiments in a rather consistent way. Nonetheless, since cosmological parameters inferred from observations are constrained within the framework of the standard $\Lambda$CDM model, they cannot---strictly speaking---be interpreted as observational confirmation for the basic assumptions underlying the concordance model itself; namely, a spatially flat, homogeneous and isotropic universe described by the Friedmann-Lemaître-Robertson-Walker (FLRW) metric. In this sense, a safer approach would be to interpret these results as a good agreement between theory and data, while at the same time test for the theoretical robustness of these assumptions when compared with observations.

The $\Lambda$CDM concordance model is a phenomenological model---and, as such, always under revision---whose scope is to fulfill the primary goal in cosmology: to find the simplest model that fit all observational data. A secondary, yet important, goal in cosmology is to explore the range of models that are compatible with all data, in order to understand if the simplest model is probable \cite{bok:EllisWainwright}. A question naturally arises: are there spatially anisotropic solutions to general relativity that cannot be distinguished (dynamically) from the FLRW cosmologies?
 
The cosmological principle~\cite{bok:bondi68}, which nowadays is taken to stipulate that the Universe is spatially homogeneous and isotropic,\footnote{In its simplest version, the cosmological principle should also postulate the Universe's spatial topology \cite{Uzan:2016wji}. Usually, and this is the case here, this is taken to be the trivial topology.} is implemented at the background level. When a linear description of the Universe is appropriate---which is usually the case at scales $\gtrsim 100{\rm Mpc}$---linear stochastic fluctuations transfer these symmetries to the statistical level. To understand whether or not these principles are enforced by observational data, it is necessary to consider models that break at least some of them, and to test its observational consequences. To this end, a natural framework is provided by \emph{Bianchi models}, which are spatially homogeneous cosmological models that may break three-dimensional rotational invariance in various ways: the fundamental congruence may possess  shear~\cite{Collins:1973lda,Bunn:1996ut,Koivisto:2008ig,Hsu-Wainwright:1986,Wainwright-Hsu:1989,Hewitt-Wainwright:1993} and vorticity~\cite{hawking1969rotation,Collins:1973lda,Barrow:1985tda,int:Obukhov1992,Bunn:1996ut,Obukhov:2000jf}; the matter sector may possess anisotropic stress \cite{Maccallum:1970eg,Barrow:1997sy,LeBlanc:1997,Kanno:2008,Watanabe:2009,Calogero:2009,Yamamoto:2011,int:Thorsrud11,int:Thorsrud12,Ben18,Normann:2019_1,Normann:2019_2, Almeida:2019iqp} and tilt~\cite{Barrow:2003fc,Coley-Hervik:2005,Sundell:2015gra}; spatial sections of homogeneity may possess anisotropic curvature \cite{barrow1984helium,Hsu-Wainwright:1986,Wainwright-Hsu:1989,Hewitt-Wainwright:1993}. Note that the various types of anisotropies do not evolve  independently but are intimately connected via Einstein's field equations. For instance, the evolution equation for the shear tensor, which measures the difference between the rates of expansion in different directions, is ``sourced'' by the tensors that describe the anisotropic stress of matter and anisotropic spatial curvature.\footnote{From section \ref{sec:Conditions}, we use $\sigma_{\mu\nu}$ for the shear tensor, $\pi_{\mu\nu}$ for the anisotropic stress tensor, and what we here refer to, informally, as \emph{anisotropic spatial curvature} is the trace-less three-dimensional Ricci tensor ${^3}S_{\mu\nu}$ as defined therein.} As Bianchi models are the most general spatially homogeneous models that include open, flat and closed FLRW models as special cases, they provide a framework where the  assumptions of the standard model, as well as interesting alternative models, may be rigorously addressed.


In this paper we devote attention to a peculiar corner of Bianchi models: orthogonal models with anisotropic spatial curvature that nevertheless expands isotropically and, therefore, possess a shear-free fluid flow like standard FLRW cosmologies. Since such solutions posses a conformal Killing vector parallel to four-velocity of matter, the isotropy of the background CMB remains intact \cite{Ehlers:1966ad,Clarkson:1999yj,Obukhov:2000jf}. For the same reason, such models are parallax-free \cite{int:Hasse88}. The expansion history is similar to FLRW models, but, at the same time, there are interesting phenomenological consequences at the perturbative level, as discussed in \cite{pert:Zlosnik11,pert:Pereira12,Pereira:2015pxa,pert:Pereira17}. The possibiliy of shear-free cosmological models with anisotropic spatial curvature were first discussed by Mimoso and Crawford in \cite{main:Mimoso93}, and has been revisited several times \cite{main:Coley94,main:Coley94nr2,shear-free:Abebe16,main:Koivisto11,Menezes:2012kc}. Recall that within spatially homogeneous models containing only non-tilted perfect fluids the metric is necessarily FLRW if the fluid flow is shear-free \cite{main:Collins82}. Consequently, as pointed out in \cite{main:Mimoso93}, shear-free orthogonal models with anisotropic spatial curvature require inclusion of imperfect matter, i.e. a fluid that possess anisotropic stress. The first exact solution with a physical matter model was presented in \cite{main:Carneiro01}, realized by a massless scalar field with an isotropy-violating gradient $\nabla_\mu \varphi$ in a Locally Rotationally Symmetric (LRS) Bianchi type III universe. More recently, this solution has been proved to be the unique shear-free solution with anisotropic spatial curvature within spatially homogeneous orthogonal models that contain, in addition to a collection of standard perfect fluids, a free canonical $p$-form gauge field with $p\in\{0,1,2,3\}$ \cite{Mik18}.\footnote{The cases $p=0$ and $p=2$ are physically
equivalent via Hodge dual at the field strength $p+1$ level.} Realizations with a Maxwell field ($p=1$) have recently been found in Bianchi type III, VIII and IX by including a charged matter field (beyond the matter model considered in \cite{Mik18}) that couples to the Faraday tensor \cite{Pailas:2019}.

\subsection{Main results and interpretation}
The scope of this paper is two-fold. First, in section \ref{sec:Conditions} we review and clarify the mathematical backbone of shear-free cosmological models. We define \emph{extended FLRW} (eFLRW) models via a matter ansatz that, in addition to the conditions given in \cite{main:Mimoso93}, assumes a third condition which implies that the anisotropic cosmological model can be mapped onto a reference FLRW model with the same expansion history. Realizations with a massless scalar field are next reviewed in section \ref{sec:realizations}. The shear-free LRS Bianchi type III solution \cite{main:Carneiro01} is showed to fit into our definition of an eFLRW model, with the same expansion history as an open FLRW model. Realizations in Bianchi type VI$_0$ and the Kantowski-Sachs metric are also reviewed, in order to present eFLRW models that can be mapped onto flat and closed FLRW models, respectively. The latter two examples require, however, a pathological scalar field that violates the positive energy condition. 

The second scope and primary goal of this paper is to investigate the stability of eFLRW models within the space of more general Bianchi models. We employ the dynamical systems approach and consider all spatially homogeneous models that contain a canonical massless scalar field with a spatially homogeneous gradient $\nabla_\mu \varphi$ and a standard non-tilted perfect fluid with equation of state $p=w\rho$. This choice of matter model picks out the exact solution in LRS Bianchi type III as the unique eFLRW model. We carry out a stability analysis within LRS Bianchi type III in section \ref{ch:LRS} and within the full space of homogeneous models in section \ref{sec:arbitraryHomPert}. Remarkably, we find the equilibrium points of the eFLRW model to be attractors and thus stable. Our main results are:
\begin{enumerate}
	\item All LRS Bianchi type III universes with space-like $\nabla_\mu\varphi$ are asymptotically shear-free, for all $w\in[-1,1]$. See Theorem \ref{thm:global} for details. 
	\item All shear-free equilibrium sets with anisotropic spatial curvature are stable with respect to all homogeneous perturbations for $w\geq -1/3$. See Theorem \ref{thm:local} for details.
\end{enumerate}

Our results may well appear surprising. They can be interpreted in the following way: matter anisotropies and geometrical anisotropies mutually rearrange, dynamically, so that they counteract each other and cancel on the right-hand side of the shear-evolution equation. Hence, assymptotically, the Bianchi III model isotropizes as a Bianchi I model containing only perfect fluids. This fits into a pattern seen in other contexts in Bianchi type I models: when multiple anisotropic matter fields are included in the cosmic mix, the fields tend to rearrange dynamically so that the Universe isotropizes. In \cite{Yamamoto:2012} the considered fields were gauge field vectors non-minimally coupled to a scalar field in the context of early Universe  inflation. A ``cosmic minimum hair conjecture'' was proposed: ``the Universe organizes itself so that any feature of the spacetime during inflation becomes minimum''. We would like to point out that the underlying mechanism seems not to be inflation per se, because, in our context the same mechanism is at action with a positive deceleration parameter $q$. This is also seen in the findings of \cite{Adamek:2011}, where anisotropies from free-streaming relativistic particles and a large scale magnetic field are shown to counteract each other such that the Bianchi type I universe isotropizes, again in a non-inflationary epoch of cosmic history. We propose that the isotropization, that evidently occurs rather generically when multiple (anisotropic) degrees of freedom are present in spatially homogeneous models, can be attributed to \emph{general relativity} and its non-linearity itself. To the best of our knowledge the non-linear analysis presented here is the first one where geometrical anisotropies mix with matter anisotropies, thereby realizing essentially the same mechanism of \emph{dynamical cancellation of cosmological anisotropies} in a more general Bianchi model.

An interesting open question is if this mechanism is unique to general relativity or if it can be attributed to the non-linearity of gravitational theory in a broader sense. Note that the shear-free condition for orthogonal cosmological models given for general relativity in \cite{main:Mimoso93} extends to $f(R)$ theories \cite{shear-free:Abebe16}, and so does the linear damping result for small shear.

\paragraph{Organization}
The rest of the paper is structured as follows. In Section \ref{sec:Conditions} we lay out the mathematical backbone of the study, giving also the conditions for eFLRW cosmologies. In Section \ref{sec:realizations}, we show how these conditions are met in a particular example with a massless scalar field. The stability of the shear-free solutions is assessed in Sections \ref{ch:LRS} and \ref{sec:arbitraryHomPert}. A conclusion is provided in Section \ref{sec:Conclusion}. Finally, an explicit calculation of photon propagation in the spacetimes studied is provided in Appendix \ref{sec:PhotonProp}. 

\paragraph{Conventions}
We adopt metric signature $(-,+,+,+)$ and units in which $c=1=8\pi G$. Greek indices $(\alpha, \beta, \dots)$ run from 0 to 3, Latin indices $(a, b, \dots)$ run from 1 to 3.

\section{Conditions for extended FLRW models}
\label{sec:Conditions}
In this section we review the class of shear-free Bianchi models, define \emph{extended FLRW models} and show how such solutions can be mapped onto a reference FLRW model with the same expansion history.  

\subsection{Shear-free Bianchi metrics}
In spatially homogeneous cosmological models the invariant line-element can always be written as
\be
ds^2 = -dt^2 +  h_{ab}(t)  w^a   w^b \,,
\label{line-element:0}
\ee
where the one-forms are functions only of spatial coordinates: $w^a=e^a_{\;\;i}(x,y,z)dx^i$. In order for this to correspond to a Bianchi model, the set $\{w^a\}$ must be closed in the sense that $\mathrm d w^a = -\tfrac{1}{2} D_{bc}^{\;\;\; a} w^b \wedge w^c$, where the commutation coefficients $D_{bc}^{\;\;\; a}$ are constants that can be decomposed as\footnote{In a standard left-invariant frame the commutation constants of the basis ($D_{bc}^{\;\;\; a}$) and the structure coefficients of the Killing vectors (usually denoted by $C_{bc}^{\;\;\; a}$) are related by an overall factor of $(-1)$, see \cite{bok:MacCallum79,bok:Hervik} for details. With the minus sign on the left-hand side of (\ref{decomposition}) our conventions agree with standard definitions.}
\be
-D_{ab}^{\;\;\;\;d}  = \epsilon_{abe} N^{ed} + A_e\left( \delta^e_a \delta^d_b - \delta^e_b \delta^d_a \right)
\label{decomposition} \,,
\ee  
and where $\epsilon_{abc}$ is the totally anti-symmetric symbol ($\epsilon_{123}=1$). The Bianchi classification consists in enumerating all covectors $A_a$ and symmetric matrices $N^{ab}$ compatible with the Jacobi identity: $D_{[ab}^{\;\;\;\;\; e} \;  D_{d]e}^{\;\;\;\;\; f}  = 0$; see \cite{bok:Ellis98,bok:Hervik} for details.

In the metric approach to Bianchi models the basic variables are $h_{ab}(t)$ which are functions of time, only. It is convenient to work in the basis $w^\mu=\{dt, w^a\}$ which is dual to $e_\mu = \{ \partial_t, e_a \}$. In this non-coordinate basis the components of the metric tensor are $g_{\mu\nu}=h_{\mu\nu}-u_\mu u_\nu$ where $h_{0\nu}=0$ and $u_\mu = -\partial_\mu t$, and the Levi-Civita connection is 
\be
\Gamma^\gamma_{\;\; \alpha\beta} = \frac{1}{2} g^{\gamma \delta} \left( \partial_\beta g_{\delta \alpha} + \partial_\alpha g_{\delta\beta} - \partial_\delta g_{\alpha\beta} + D_{\delta\alpha\beta}+ D_{\delta\beta\alpha} - D_{\alpha\beta\delta} \right) \,,
\label{Levi-Civita}
\ee 
where $\partial_\mu$ is the directional derivative along the basis vector $e_\mu$.

In Bianchi models there are hypersurfaces of homogeneity, labelled $\Sigma_t$, where $t=\text{constant}$. The \emph{normal congruence} is given by the unit time-like vector field $u^\mu$ which is orthogonal to $\Sigma_t$. This family of observers is necessarily geodesic and irrotational and it follows that its covariant derivative has the following irreducible representation:
\be
\nabla_\mu u_\nu = H h_{\mu\nu} + \sigma_{\mu\nu} \,,
\label{irred}
\ee  
where the scalar $H(t)$ is the Hubble parameter and $\sigma_{\mu\nu}(t)$ is the shear tensor, which is symmetric $\sigma_{\mu\nu}=\sigma_{(\mu\nu)}$, spatial $\sigma_{\mu\nu}u^\nu=0$, and trace-free $\sigma^\mu_{\;\;\mu}=h^{\mu\nu}\sigma_{\mu\nu}=0$. In this work we are interested in the case where the normal congruence $u^\mu$ is shear-free. Setting $\sigma_{\mu\nu}=0$ in (\ref{irred}) and using (\ref{Levi-Civita}) results in 
\be
h_{ab}(t)=\left(e^{\int^t H(t') dt'} \right)^2\,\hat h_{ab} \,, 
\label{hei}
\ee
where $\hat h_{ab}$ are integration constants. At this step it is convenient to define a scale factor $a(t)$ by $H=\dot a / a$, where a dot denotes differentiation with respect to time $t$, so that (\ref{hei}) can be written $h_{ab}(t)=a^2(t) \hat h_{ab}$.  To summarize, in the case of a shear-free normal congruence the line-element (\ref{line-element:0}) reduces to 
\be
ds^2 = -dt^2 + a^2(t) \hat h_{ab}  w^a   w^b \,,
\label{line-element2}
\ee
where the constants $\hat h_{ab}$ form a symmetric and positive definite matrix.\footnote{In order to span all Bianchi geometries one must allow $\hat h_{ab}$ to have off-diagonal elements. A convenient procedure, based on the Cauchy-Schwarz inequality, to eliminate false solutions ($\hat h_{ab}$ not positive definite) was introduced in section 5.1 of \cite{Mik18}.} It is convenient to define $\det(h_{ab})=a^6(t)$ so that $\hat h_{ab}$ has five independent constant components and unity determinant. Then, the only dynamical degree of freedom is the scale factor $a(t)$, which controls distances in $\Sigma_t$. We note that the shear-free line-element (\ref{line-element2}) do not imply any restriction on the space of three-dimensional geometries, which is spanned by $\hat h_{ab}$ rather then $h_{ab}(t)$. But in the shear-free case the geometry is necessarily ``frozen'', in the sense that the expansion corresponds only to conformal transformations of hypersurfaces $\Sigma_t$.

\subsection{Energy-momentum tensor}
Next we turn the attention to the matter sector. Let us start by decomposing an arbitrary energy-momentum tensor $T_{\mu\nu}$ relative to the normal congruence $u^\mu$ \cite{bok:Ellis98}:
\be
T_{\mu\nu} =\rho u_\mu u_\nu+ p  h_{\mu\nu}+\pi_{\mu\nu}+2q_{(\mu}u_{\nu)}\,.
\label{EMT:0a}
\ee
An observer with four-velocity $u^\mu$ sees the following energy density $\rho$, pressure $p$, anisotropic stress $\pi_{\mu\nu}$, and energy flux $q^\gamma$:
\be
	\rho=u^\mu u^\nu T_{\mu\nu}\,,\;\;
	p = \frac{1}{3}h^{\mu\nu}T_{\mu\nu}\,,\;\;
	\pi_{\mu\nu} = \left( h_{(\mu}^{\;\;\;\alpha} h_{\nu)}^{\;\;\;\beta} - \frac{1}{3}h_{\mu\nu}h^{\alpha\beta} \right) T_{\alpha\beta} \,,\;\;
	q^{\gamma} = -h^{\gamma \mu} u^\nu T_{\mu\nu}\,.
\label{EMT:0b}
\ee
Note that $\pi_{\mu\nu}$ is symmetric, spatial and trace-free, and that $q^\gamma$ is spatial.

\paragraph{Perfect fluids} Firstly, we consider a collection of perfect fluids which are all non-tilted in the sense that they are comoving with the normal congruence $u^\mu$:   
\be
T_{(\ell)}^{\mu\nu} = \rho_{\ell} u^\mu u^\nu + p_{\ell} h^{\mu\nu} \,,  
\label{EMT:1}
\ee
where the index $\ell$ labels each of the perfect fluids. Since the normal congruence is shear-free by construction, each of these fields obey a shear-free fluid flow, like in the standard $\Lambda$CDM cosmology. Now, if the matter sector is restricted to these perfect fluids, the only possible cosmological models with the metric (\ref{line-element2}) are FLRW models, which are subsets found in Bianchi types I and VII$_0$ (flat), V and VII$_h$ (open) and IX (closed).

\paragraph{Imperfect fluid} In order to construct extended FLRW models we will also need an imperfect fluid, denoted by '$x$', that allows the spatial anisotropies to freeze out so that the conformal expansion remains intact, in accordance with the line element (\ref{line-element2}). We require the associated energy-momentum tensor $\tau_{\mu\nu}$ to have the following properties:
\begin{enumerate}
	\item The anisotropic stress $\pi_{\mu\nu}$, with respect to the normal congruence $u^\mu$, is identical to the anisotropic spatial curvature tensor ${^3}S_{\mu\nu}$, defined as the trace-free 3-dimensional Ricci tensor on hypersurfaces of homogeneity:
	\be
	\pi_{\mu\nu} = {^3}S_{\mu\nu}, \quad {^3}S_{\mu\nu} \equiv {^3}R_{\mu\nu} - \frac{1}{3} {^3}R h_{\mu\nu}\,.
	\label{ansatz:1}
	\ee
	\item There is no energy flux in hypersurfaces of homogeneity:
	\be
	q_\mu=0 \,.
	\label{ansatz:2}
	\ee
	\item The Raychudhuri equation is sourced exclusively by perfect fluids such that:
	\be
	\rho_x + 3p_x = 0 \,.
	\label{ansatz:3}
	\ee
\end{enumerate}
The three assumptions (\ref{ansatz:1})-(\ref{ansatz:3}) above will later be referred to as the \emph{matter ansatz} for the imperfect fluid. The first and second conditions are necessary and sufficient conditions for general relativistic solutions with the metric (\ref{line-element2}). They are equivalent to the conditions given by Mimoso and Crawford in \cite{main:Mimoso93}. The third condition implies that the solution can be mapped onto a reference FLRW model with the same background evolution history, as shown in the subsection below. In fact this ``extra'' third condition is often, depending on the details of the matter model, a direct consequence of the first one. In order to see this, note that (\ref{ansatz:1}) implies that the spatial components of $\tau_{\mu\nu}$ decay as $\pi^{\mu}_{\;\;\nu} \propto 1/a^2(t)$.\footnote{Here the components are assumed relative to an orthonormal frame.} It is therefore quite natural that all the components of $\tau^{\mu}_{\;\;\nu}$ decay uniformly as $1/a^2(t)$, which is what the third condition (\ref{ansatz:3}) assumes.\footnote{Since the normal congruence is shear-free the conservation equation takes the usual form $\dot \rho_x + 3H(\rho_x+p_x)=0$. Therefore (\ref{ansatz:3}) implies $\rho_x\propto 1/a^2(t)$ and $p_x\propto 1/a^2(t)$.} This is the case for a massless scalar field, as will be seen in section \ref{sec:realizations}.

The total energy-momentum tensor can then be written as
\be
T^{\mu\nu} = \sum_\ell T_{(\ell)}^{\mu\nu} + \tau^{\mu\nu} = \left( \rho_{x} + \sum_\ell \rho_{\ell} \right) u^\mu u^\nu + \left( -\frac{\rho_{x}}{3} + \sum_\ell p_{\ell} \right) h^{\mu\nu} + {^3}S^{\mu\nu} \,.  
\label{EMT:2}
\ee
For later reference, we write down a formal definition of an extended FLRW model:
\begin{definition}[eFLRW model]
An extended FLRW model is a general relativistic solution with a line-element that can be written on the form (\ref{line-element2}) and with a total energy-momentum tensor that can be written on the form (\ref{EMT:2}) with nonvanishing ${^3}S^{\mu\nu}$.  
\label{def:eFLRW}
\end{definition}

\subsection{Reference FLRW model}
Given the conditions above it follows that the energy density of the imperfect fluid decays in the same way as spatial curvature:  $\rho_{x}(t)\propto 1/a^2(t)$ and ${^3}R(t) \propto 1/a^2(t)$. This suggests that we introduce an effective curvature constant as \cite{Mik18}
\be
\keff \equiv \frac{a^2(t)}{6}  \cdot \left({^3}R(t) - 2\rho_{x}(t) \right)\,,  
\label{keff}
\ee  
which is constant in time (and space) for all eFLRW models. With the metric (\ref{line-element2}) and energy-momentum tensor (\ref{EMT:2}) Einstein's field equation reduces to the Friedmann and Raychudhuri equations, which take exactly the same form as in FLRW models:
\begin{align}
	&H^2 + \frac{\keff}{a^2} = \frac{\rho}{3} \,, \label{Friedmann} \\
	&\dot H + H^2 = -\frac{1}{6} (\rho+3p) \,, \label{Raychudhuri}
\end{align}
where $\keff$ is a constant and
\[
\rho = \sum_{\ell} \rho_\ell \quad \text{and} \quad p = \sum_{\ell} p_\ell  \,. 
\]
Note that the index $\ell$ runs only over perfect fluids as a consequence of the third condition (\ref{ansatz:3}) of the matter ansatz. Thus $\rho$ and $p$ represent total energy density and pressure of the collection of perfect fluids, which can be associated with standard $\Lambda$CDM matter fields like cold dark matter, radiation and ordinary matter. 

To conclude, all eFLRW models can be mapped onto a reference FLRW model with the same expansion history, $a(t)$, in the following way:\footnote{Since the domain (eFLRW models) and codomain (FLRW models) are subsets of orthogonal Bianchi models, $q_\mu=0$ for all matter fields.}
\be
(H,\;\rho_\ell,\;p_\ell,\;\rho_x,\;{^3}R,\;{^3}S_{\mu\nu}) \longmapsto (H,\;\rho_\ell,\;p_\ell,\;0,\;{^3}R-2\rho_x,\;0) \,.
\label{map:0}
\ee
We call the reference FLRW model \emph{open} if $\keff < 0$, \emph{flat} if $\keff = 0$ and \emph{closed} if $\keff > 0$. Assuming the positive energy condition $\rho_x>0$, it follows from (\ref{keff}) that a flat or closed reference model is possible only in Bianchi type IX and the Kantowski-Sachs metric because ${^3}R \le 0$ in Bianchi type I-VIII \cite{bok:Hervik}.

In FLRW models $\keff$ is of course identified with the curvature constant $k$ of the FLRW metric:
\be
ds^2 = -dt^2 + a^2(t) \left( \frac{dr^2}{1-k r^2} + r^2 (d\theta^2 + \sin^2\theta \; d\phi^2) \right)\,.
\label{metric:FLRW}
\ee
In this case ${^3}S_{\mu\nu}=0$ and consequently all matter fields are perfect fluids, in accordance with the energy-momentum tensor (\ref{EMT:2}). Conversely, if the metric (\ref{line-element2}) possess anisotropic spatial curvature (${^3}S_{\mu\nu}\neq0$), then an imperfect fluid is required. These are the solutions to which we refer as extended FLRW models. 

\section{Realizations with a scalar field}
\label{sec:realizations}
A free scalar field possess an energy-momentum tensor capable of fulfilling the properties required by the ansatz (\ref{ansatz:1})-(\ref{ansatz:3}), yet rich enough to break these conditions. It thus provides a plausible physical model which allows us to investigate the stability of these conditions without any ad-hock assumptions.

Thence, from here on the generic imperfect matter field, denoted above by $x$, is associated with $\varphi$; namely, a free massless scalar field with Lagrangian density
\be
\mathcal L_\varphi = -\frac{1}{2} \nabla_\gamma \varphi \nabla^\gamma \varphi \,,
\label{L:phi}
\ee
and energy-momentum tensor
\be
\tau_{\mu\nu} \equiv -\frac{2}{\sqrt{-g}} \frac{\delta}{\delta^{\mu\nu}}\left(\sqrt{-g} \mathcal L_\varphi \right) = \nabla_\mu \varphi \nabla_\nu \varphi -\frac{1}{2} g_{\mu\nu} \nabla_\gamma \varphi \nabla^\gamma \varphi \,.
\label{EMT:3}
\ee
In order to see how this provides the required ``imperfect fluid'', it is instructive to decompose its gradient relative to the normal congruence $u^\mu$ as:
\be
\nabla_\mu \varphi = -\vartheta u_\mu + v_\mu \,,
\ee
where $v^2=v^\gamma v_\gamma>0$ and $u^\gamma v_\gamma=0$ so that $v^\alpha$ is a spacelike vector orthogonal to $u^\alpha$. Since $\nabla_\mu \varphi$ must be spatially homogeneous we require $\vartheta=\vartheta(t)$ and $v^\alpha=v^\alpha(t)$. The energy density, pressure, energy flux and anisotropic stress, as defined by (\ref{EMT:0b}), can now be written as
\be
\rho_\varphi =\frac{1}{2} (v^2 + \vartheta^2)\,,  \quad
p_\varphi = \frac{1}{2} \left(-\frac{1}{3}v^2+\vartheta^2\right)\,,\quad
q^\mu =-\vartheta v^\mu\,, \quad 
\pi_{\mu\nu} = v_\mu v_\nu - \frac{1}{3} v^2 h_{\mu\nu} \,. 
\label{EMT:4}
\ee
If $\nabla_\mu \varphi$ is orthogonal to the normal congruence (i.e., $\vartheta=0$), it follows that $\rho_\varphi + 3p_\varphi = 0$ and $q^\mu=0$, which amounts to conditions (\ref{ansatz:2}) and (\ref{ansatz:3}). Therefore, eFLRW models can be physically realized by a free scalar field.\footnote{The realization using a 2-form gauge field in \cite{main:Koivisto11} is equivalent to a massless scalar field ($\varphi$) upon Hodge dual at the field strength level ($\nabla_\mu\varphi$) \cite{Mik18}.} 

Yet, among metrics of the type (\ref{line-element2}), the unique shear-free solution with anisotropic spatial curvature (${^3}S_{\mu\nu}\neq0$) is found in LRS Bianchi type III with trace-free matrix $N^{ab}$. This result was proved by one of us in \cite{Mik18}. A coordinate representation of this solution --- first found by Carneiro and Marug\'an \cite{main:Carneiro01} --- is given in Table \ref{tab:shear-free}. Note that $\keff<0$ so its expansion history is similar to an open FLRW model. The spatial geometry corresponds to a product between the maximally symmetric negatively curved 2-space $\mathcal{H}_2$ and the real line $\mathcal{R}_1$, where the LRS axis is associated with the latter.\footnote{In fact all Bianchi type III geometries with trace-free $N^{ab}$ are $\mathcal{H}_2\times \mathcal{R}_1$ and thus LRS spaces, see \cite{Mik18} for details. LRS Bianchi type III spaces with $N^{a}_{\;\;a}\neq 0$, on the other hand, are identical to LRS Bianchi type VIII spaces \cite{bok:MacCallum79}.}  

If one is willing to forfeit the weak energy condition by flipping the sign of the Lagrangian density (\ref{L:phi}), one can also construct extended FLRW models in Bianchi type VI$_0$ (with trace-free $N^{ab}$) and the Kantowski-Sachs metric \cite{Mik18}.\footnote{The Kantowski-Sachs metric \cite{bok:KantowskiSachs66, bok:Kantowski66} is the unique case of a spatially homogeneous spacetime in which the isometries do not admit a 3-dimensional group that acts simply transitively on hypersurfaces of homogeneity, and therefore falls outside the Bianchi classification \cite{bok:Hervik}.} See Table \ref{tab:shear-free} for coordinate representations. Here, the appearance of the imaginary unit $i$ in $\varphi$ reminds us that the weak energy condition is violated and that $\mathcal L_\varphi\rightarrow -\mathcal L_\varphi$ in (\ref{L:phi}) is needed. We stress that these solutions are included here merely to show examples of extended FLRW models that are mapped onto flat and closed FLRW models.

Figure \ref{Fig1} illustrates the dynamics of the three extended FLRW models in Table \ref{tab:shear-free}, as well as their maps onto the reference FLRW model. Here we used 
expansion normalized variables $\Omega_k$ and $\Omega_\varphi$ defined as
\be
\Omega_{k} = -\frac{{^3}R}{6H^2} \,, \qquad \Omega_{\varphi} = \frac{\rho_\varphi}{3H^2} \,.
\label{normalized:1}
\ee
These variables are subject to the constraint 
\be
\Omega_\varphi=r \Omega_k\,, \quad r=\text{constant}\,,
\label{def:r}
\ee
since both ${^3}R$ and $\rho_\varphi$ evolve proportionally to $1/a^2(t)$. 
The value of $r$ is $1/2$ for Bianchi type III and Kantowski-Sachs and $-1$ for Bianchi type VI$_0$, as follows from the information presented in Table \ref{tab:shear-free}. In all cases the time evolution is given by   
\be
(\Omega_{k})' = \frac{\rho+3p}{\rho}  (1-\Omega_{k}-r\Omega_{k})\Omega_{k} \,,
\label{ev:Okeff}
\ee
where the prime denotes differentiation with respect to $\ln a$. Note that this evolution equation also holds for FLRW models when $r=0$.
\begin{table}[h]
	\footnotesize
	\newcommand\T{\rule{0pt}{2.6ex}}
	\newcommand\B{\rule[-1.2ex]{0pt}{0pt}}
	\begin{tabular}{ l l l l l l l l l l }
		& & & & \\
		& & & & & \\
		\hline\hline \B
		Type \T \B & ref. FLRW  & $w^1$  &  $w^2$  &  $w^3$ & ${^3}R$ & $\varphi$ & $\rho_\varphi$ & $\keff$    \\ \hline \T
		\T\B III & open & $dx$ & $e^{kx}dy$ & $dz$ & $-2k^2/a^2(t)$ & $kz$ & $k^2/2a^2(t)$ & $-k^2/2$   \\			
		\T\B VI$_0$ & flat  & $e^{-kz}dx$ & $e^{kz}dy$ & $dz$ & $-2k^2/a^2(t)$ & $i\sqrt{2}kz$ & $-k^2/a^2(t)$ & $0$    \\				
		\T\B K.S. & closed & $dx$ & $k^{-1} \sin(k x)dy$ & $dz$ & $2k^2/a^2(t)$ & $ikz$ & $-k^2/2a^2(t)$ & $k^2/2$    \\		
		\hline\hline
	\end{tabular}
	\caption{Coordinate representation of eFLRW solutions in LRS Bianchi type III ($N^a_{\;\;a}=0$), pseudo-LRS Bianchi type VI$_0$ ($N^a_{\;\;a}=0$) and the Kantowski-Sachs metric. The invariant line-element (\ref{line-element2}) is given by the basis one-forms $w^a$ and the matrix $\hat h_{ab}=\text{diag}(1,1,1)$. \newline} 	
	\label{tab:shear-free}
\end{table}
\newline
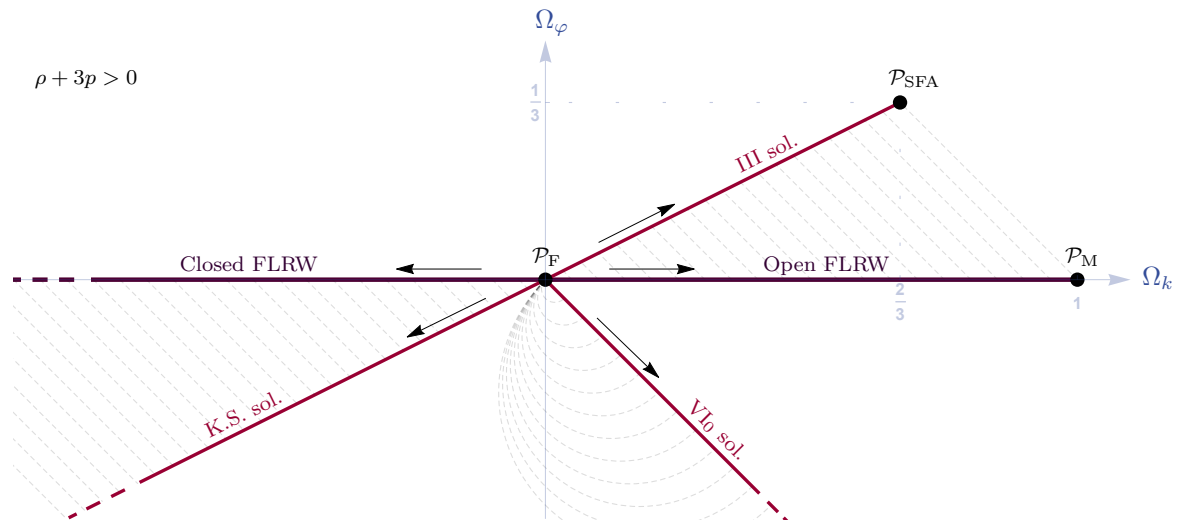
\begin{figure}[h]
	\centering	
	\begin{overpic}[width=0.9\textwidth,tics=10]{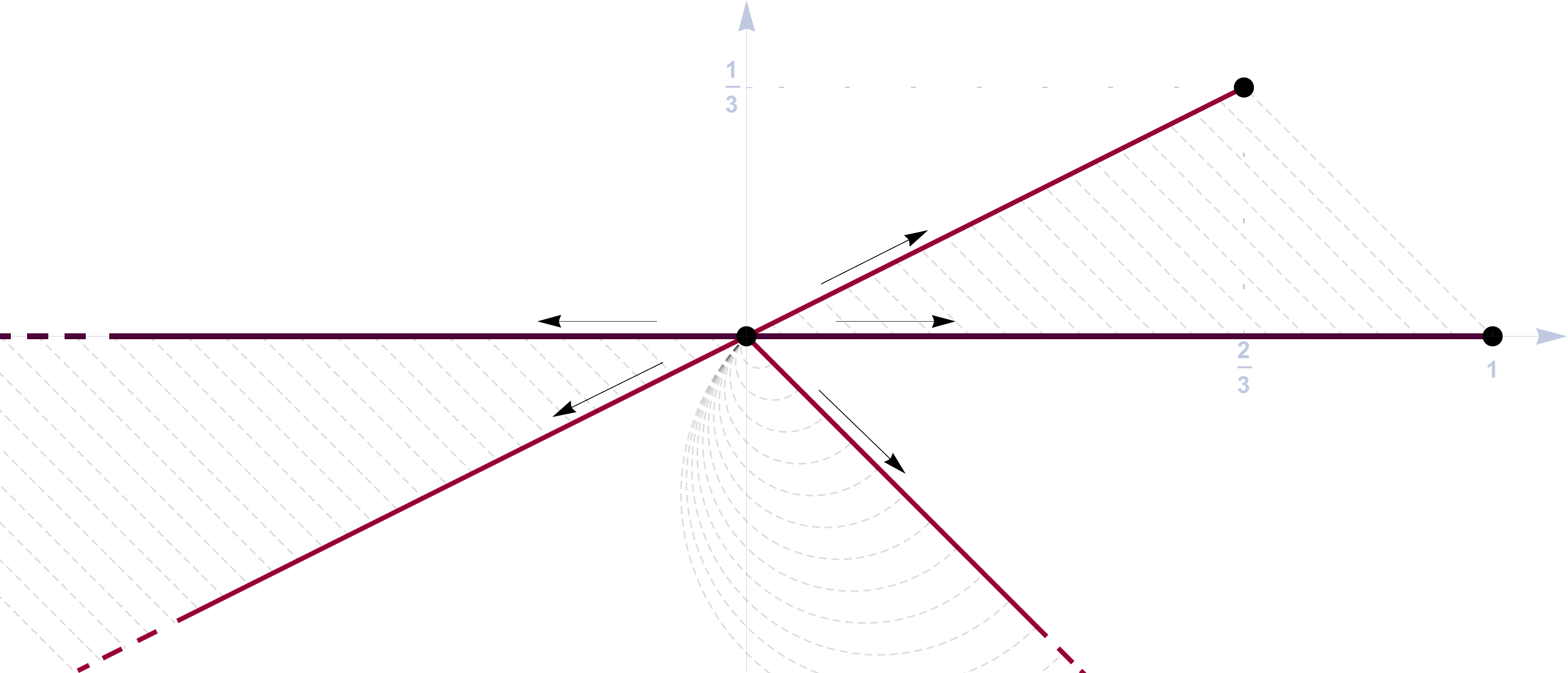}		
		\put (46.7,44) {\footnotesize \textcolor{coordcolor}{$\displaystyle \Omega_{\varphi}$}}
		\put (101,21) {\footnotesize \textcolor{coordcolor}{ $\displaystyle \Omega_k$}}
		\put (67.1,22.2) {\scriptsize \textcolor{FLRWcolor}{Open FLRW}}
		\put (14.9,22.2) {\scriptsize \textcolor{FLRWcolor}{Closed FLRW}}
		\put (64.5,31) {\scriptsize \textcolor{sfcolor}{\rotatebox{26}{III sol.}}}
		\put (60,10) {\scriptsize \textcolor{sfcolor}{\rotatebox{-45}{VI$_0$ sol.}}}
		\put (17,7.3) {\scriptsize \textcolor{sfcolor}{\rotatebox{26}{K.S. sol.}}}
		\put (93.2,23) { \scriptsize $\displaystyle\mathcal{P}_{\rm M}$ }
		\put (45.5,23) { \scriptsize$\displaystyle\mathcal{P}_{\rm F}$}
		\put (77.38,39) { \scriptsize $\displaystyle\mathcal{P}_{\rm SFA}$ }
		\put (1,39) { \scriptsize $\displaystyle \rho + 3p>0$ }
	\end{overpic}	
	\caption{Illustration of three eFLRW solutions: LRS Bianchi type III, Bianchi type VI$_0$ and Kantowski-Sachs. The map (\ref{map}) onto the reference FLRW model are given by gray dashed lines. The arrows show the direction of evolution when $\rho+3p>0$; they point towards $\mathcal P_\text{F}$ (flat FLRW) when $\rho+3p<0$.} 
	\label{Fig1}
\end{figure}

As in FLRW models, the direction of evolution is determined by the sign of the combination $\rho+3p$, where we recall that $\rho$ and $p$ are sums exclusively over the perfect fluids. Figure \ref{Fig1} also illustrates the map onto the reference FLRW model, given by
\be
(\Omega_k\,,\; \Omega_\varphi) \longmapsto (\Omega_k+\Omega_\varphi\,,\; 0) \,,
\label{map}
\ee
and in accordance with (\ref{map:0}). We thus see that LRS Bianchi type III, Bianchi type VI$_0$ and the Kantowski-Sachs metrics are mapped onto open, flat and closed reference FLRW models, respectively. Also note from the figure that the positive energy condition $\Omega_\varphi>0$ is realized only in the Bianchi type III model, as discussed above. In this case the eFLRW model evolves in the direction toward the point $\mathcal P_\text{SFA}$ (later referred to as the ``shear-free attractor'' for reasons established in section \ref{ch:LRS} and \ref{sec:arbitraryHomPert}) when $\rho+3p>0$ and toward the point $\mathcal P_\text{F}$  (flat FLRW) when $\rho+3p<0$. In particular note that $\mathcal P_\text{SFA}$ is mapped onto the Milne solution $\mathcal P_\text{M}$, which is the spacetime (\ref{metric:FLRW}) with $k=-1$ and $a(t)=t$.

\section{Stability of shear-free solutions \label{ch:LRS}}
We have seen that extended FLRW models can be realized by a scalar field $\varphi$. This requires that its energy-momentum tensor satisfies the ingredients (\ref{ansatz:1})-(\ref{ansatz:3}). In this section we address the issue of initial conditions by employing the dynamical systems approach. We shall restrict the attention to LRS models and focus exclusively on Bianchi type III, since this is the unique spacetime that admits an extended FLRW model realized by a scalar field respecting the weak energy condition.

A general LRS Bianchi type III spacetime can be described by the metric:
\be
ds^2 = -dt^2 +  e^{2\alpha(t)} \left[ e^{2\beta(t)}(dx^2 + e^{2kx} dy^2) + e^{-4\beta(t)} dz^2 \right]\,.
\label{LRSBIII}
\ee
In this spacetime Einstein's equation and the Klein-Gordon equation for the scalar field allow only two possibilities: $\varphi=\varphi(t)$ or $\varphi=\Phi z$, where $\Phi=\text{constant}$. As we have shown, the first case corresponds to a perfect fluid and thus it cannot realize an extended FLRW model. We are thus interested in the latter case in which $\nabla_\mu \varphi$ is spatially homogeneous, orthogonal to the normal congruence and aligned with the LRS axis. It follows from (\ref{EMT:4}) that there is zero energy flux in the hypersurfaces of homogeneity ($q_\mu=0$). Consequently, in LRS Bianchi type III with a scalar field, the three conditions (\ref{ansatz:1})-(\ref{ansatz:3}) of the matter ansatz reduce to a single relation between the spatial curvature ${^3}R$ and the energy density of the scalar field $\rho_\varphi$, as we shall verify in a moment. 

As for the collection of perfect fluids, it is at this step more tractable to consider a single component. We assume it obeys a linear barotropic equation of state 
\be
p=w\rho \,,
\ee
where $w\in[-1,1]$ is a constant. It is convenient to express Einstein's equations in terms of three dimensional curvature ${^3}R=-2k^2e^{-2\alpha-2\beta}$ and the energy density of the scalar field, $\rho_\varphi = \tfrac{1}{2}\Phi^2 e^{-2\alpha + 4\beta}$:
\begin{align}
H^2 -\sigma^2+\frac{{^3}R}{6} &= \frac{\rho+\rho_\varphi}{3}\,, \label{LRSBIII:H}\\
\dot H + H^2 &= -2\sigma^2 -\frac{\rho}{6}  (1+3w)\,,\\
\dot \sigma + 3H\sigma &= -\frac{1}{6}( {^3}R+4\rho_\varphi )\,,
\label{LRSBIII:shear} 
\end{align}
where $H=\dot\alpha$ is the Hubble parameter and $\sigma=\dot\beta$ is the shear variable. Observe from the shear evolution, equation (\ref{LRSBIII:shear}), that  
\be
{^3}R+4\rho_\varphi = 0
\label{LRSBIII:shearfree}
\ee
is a sufficient condition for the shear to remain zero. 

\subsection{Dynamical system \label{sub:DS:LRS}}
The extended FLRW model corresponds to two conditions: (\ref{LRSBIII:shearfree}) and $\sigma=0$. This represents an invariant subset under time evolution. In order to investigate the stability of these conditions we start by introducing Hubble normalized variables according to 
\be
\Sigma=\frac{\sigma}{H} \,, \quad \Omega_{k} = -\frac{{^3}R}{6H^2} \,, \quad \Omega_{\varphi} = \frac{\rho_\varphi}{3H^2} \,, \quad \Omega = \frac{\rho}{3H^2} \,.
\label{normalized}
\ee
According to the Einstein equation, these variables satisfy the constraint
\be
\Sigma^2 + \Omega_k + \Omega_\varphi + \Omega = 1 \,.
\label{Friedmann:2}
\ee
This eliminates the Hubble parameter $H$ from the dynamical system, so that its equilibrium points correspond to non-stationary universes. Switching to the scale $\alpha$ as the time parameter and eliminating $\Omega$ by means of (\ref{Friedmann:2}), we obtain the following autonomous system:
\begin{align}
	\Sigma' &= (q-2)\Sigma + \Omega_k - 2\Omega_\varphi \,, \label{DS:S} \\
	\Omega_k' &= 2\Omega_k\left( q-\Sigma \right) \,, \label{DS:k} \\
	\Omega_\varphi' &= 2\Omega_\varphi \left( q+2\Sigma \right) \,, \label{DS:p}
\end{align}
where a prime denotes differentiation with respect to $\alpha$ and
\be
q\equiv-1-\frac{\dot H}{H^2}=2\Sigma^2+\frac{1}{2}(1+3w)(1-\Sigma^2 - \Omega_k - \Omega_\varphi)
\label{q}
\ee
is the deceleration parameter. These variables are subject to the following constraints:
\be
\Sigma^2+\Omega_k+\Omega_\varphi \le 1 \,, \quad 0 \le \Omega_k  \,, \quad 0 \le \Omega_\varphi \,,
\label{existence}
\ee
which are required by $\rho\ge0$, ${^3}R\le 0$ and $\rho_\varphi\ge0$, respectively. State space is thus closed and bounded. Note that $\Omega_k>0$ corresponds to LRS Bianchi type III and $\Omega_k=0$ to flat FLRW universe. 

\subsection{Shear-free attractors and global stability}
Coordinates of state space $\mathbb R^3$ will be referenced relative to the basis 
\be
\{ \Sigma\,, \Omega_k\,, \Omega_\varphi \} \,.
\label{statespace}
\ee
Shear-free solutions lay along the parametrized curve\footnote{The parameter $\lambda$ is squared in order to match the variables used in section \ref{sec:arbitraryHomPert}; see (\ref{topdown}) and (\ref{eq:PSFA}).} 
\be
\mathcal P_\text{SF}(\lambda) = \left( 0\,, \frac{2\lambda^2}{3}\,, \frac{\lambda^2}{3} \right) \,, \quad \lambda\in[0,1]\,.
\label{curve:3D}
\ee
The underlying geometry corresponds to Bianchi type I for $\lambda=0$ and LRS Bianchi type III for $\lambda \in (0,1]$. We will refer to $\mathcal P_\text{SF}(\lambda)$ as the \emph{shear-free curve}, which is the union of flat FLRW ($\lambda=0$) and extended FLRW ($\lambda\in(0,1]$).

It is easy to check that $\mathcal P_\text{SF}(\lambda)$ is an invariant subset of state space. It evolves according to 
\be
\lambda' = \frac{1}{2}(1+3w)\lambda(1-\lambda^2) \,.
\label{ev:lambda}
\ee
Note that $\mathcal P_\text{SF}(\lambda)$ is a one-parameter family of equilibrium points when the perfect fluid has equation of state $w=-1/3$. For $w\neq -1/3$ there are two shear-free equilibrium points. The first point corresponds to parameter values $\lambda=0$, flat FLRW, and will be denoted by     
\be
\mathcal P_\text{F} = \mathcal P_\text{SF}(0) \,.
\ee
The second point corresponds to parameter value $\lambda=1$ and will be denoted by
\be
\mathcal P_\text{SFA} = \mathcal P_\text{SF}(1) \,.
\ee
We will refer to it as the \emph{shear-free attractor} (SFA) since it is globally stable for $w>-1/3$, as the analysis below will show. 

Remarkably, all universes described by the dynamical system (\ref{DS:S})-(\ref{DS:p}) with $\Omega_\varphi>0$ are asymptotically shear-free and therefore approaches a point on $\mathcal P_\text{SF}(\lambda)$ for any equation of state parameter $w\in[-1,1]$. This follows by noticing that the system (\ref{DS:S})-(\ref{DS:p}) possesses a monotonic function:
\be
\mathcal{M}' = 6q\mathcal{M}\,,\qquad \mathcal{M} \equiv \Omega_k^2 \Omega_\varphi\,.
\ee
The following theorem results:
\begin{theorem}
	All LRS Bianchi type III universes containing a scalar field with Lagrangian density $-\tfrac{1}{2}\nabla_\mu \varphi \nabla^\mu \varphi<0$ and a non-tilted perfect fluid with a linear barotropic equation of state, $p=w \rho$, are asymptotically shear-free. Specifically:
	\begin{enumerate}
		\item[a)] The shear-free attractor $\mathcal P_\text{SFA} = \mathcal P_\text{SF}(1)$ is a global attractor for $w>-\tfrac{1}{3}$. 
		\item[b)] The shear-free curve $\mathcal P_\text{SF}(\lambda)$ is a global attractor for $w=-\tfrac{1}{3}$. 
		\item[c)] The flat FLRW solution $\mathcal P_\text{F}=\mathcal P_\text{SF}(0)$ is a global attractor for $w<-\tfrac{1}{3}$. 
	\end{enumerate}   
	\label{thm:global}
\end{theorem}

\begin{proof}
	In Bianchi type III the assumption of a space-like gradient $-\tfrac{1}{2}\nabla_\mu \varphi \nabla^\mu \varphi<0$ together with Einstein equations imply $u^\mu \nabla_\mu \varphi = 0$ \cite{Mik19}, where $u^\mu$ is the normal congruence. Therefore the dynamical system (\ref{DS:S})-(\ref{DS:p}) with $\Omega_\varphi>0$ applies to the assumptions stated in the theorem. It follows from (\ref{q}) and (\ref{existence}) that $\mathcal{M}$ is monotonically increasing for $w\ge-1/3$. Since it is also bounded (state space is compact) it follows that $q\rightarrow 0$ at late times. Then, according to (\ref{q}), $\Sigma\rightarrow 0$, $\Omega \rightarrow 0$ at late times for $w > -1/3$ and $\Sigma\rightarrow 0$ at late times for $w=-1/3$. Thence, $\Omega_k - 2\Omega_\varphi \rightarrow 0$ for $w\ge -1/3$ according to (\ref{DS:S}). This is summarized by the global results stated in the theorem for $w \ge -1/3$. For $w<-1/3$ the no-hair theorem proved for this class of models in \cite{Ben18} applies. 
\end{proof}
The corresponding self-similar solutions, that represent the asymptotic future limit, are 
\be
    ds^2 = -dt^2 + \frac{t^2}{2}\left[ dx^2 + e^{2x} dy^2 + dz^2 \right]\,, \quad \varphi = z\,, \quad  \rho = 0 \,,
\label{self-similar:1}
\ee
for $w>-1/3$; 
\be
    ds^2 = -dt^2 + \frac{t^2}{2}\left[ dx^2 + e^{2\lambda x}dy^2 + dz^2 \right]\,, \quad \varphi = \lambda z\,, \quad \rho = \frac{3}{t^2}(1-\lambda^2) \,,
\label{self-similar:2}
\ee
with constant parameter $\lambda \in [0,1]$ for $w=-1/3$; and 
\be
    ds^2 = -dt^2 + t^{\tfrac{4}{3(1+w)}}  \left[ dx^2 + dy^2 + dz^2 \right]\,, \quad \varphi = 0\,, \quad \rho = \frac{4}{3(1+w)^2 t^2} \,,
\label{self-similar:3}
\ee
for $w<-1/3$.
Here non-essential parameters are omitted, such as the constant $k$ in table \ref{tab:shear-free} that can be set to unity by a coordinate transformation.\footnote{Initial conditions for ${^3}R$ are then adjusted by fixing the parameter $t$ at that instant.}

\section{Arbitrary homogeneous perturbations \label{sec:arbitraryHomPert}}

Above we considered LRS Bianchi type III and established $\mathcal P_\text{SFA}=\mathcal P_\text{SF}(1)$ as a global attractor for $w>-1/3$ and $\mathcal P_\text{SF}(\lambda)$ as a global attractor for $w=-1/3$. In this section the goal is to establish their stability, locally, with respect to all homogeneous perturbations within the space of spatially homogeneous cosmological models that contain a massless scalar field and a non-tilted perfect fluid (obeying again $p=w\rho$). By a 'homogeneous perturbation', we mean any infinitesimal change ($\delta X$) in spatial curvature, the shear tensor, the scalar field and/or the perfect fluid allowed by the constraint equations of the dynamical system. That is, if $X$ is a state vector in a spatially homogeneous model, then so is $X+\delta X$. The main result is Theorem \ref{thm:local}, which we illustrate in Figure \ref{Fig2}.


This section builds on the results of \cite{Ben18} and \cite{Mik19}, where the orthonormal frame formalism \cite{bok:EllisWainwright} was employed to investigate the structure of the space of Bianchi type I-VII$_h$ cosmological models with our considered matter sector. In what follows we adopt the notation and definitions of \cite{Mik19}. 

Consider a dynamical system where the basic variables are the expansion normalized quantities: 
\be
N_{ab}=\frac{n_{ab}}{H}\,, \quad A_a=\frac{a_a}{H}\,, \quad \Sigma_{ab} = \frac{\sigma_{ab}}{H}\,, \quad X_\mu = \frac{\nabla_\mu \varphi}{\sqrt{6}H}\,, \quad \Omega=\frac{\rho}{3H^2} \,. 
\ee
Let $d$ denote the dimension the system, after the orientation of the spatial frame has been fixed uniquely.\footnote{See section 3.3.2 of \cite{Mik19} for details.} For the discussion we need the  following sets which are invariant under time evolution:
\begin{enumerate}
	\item $\mathcal H$: all spatially homogeneous cosmological models containing a massless scalar field with a spatially homogeneous gradient $\nabla_\mu \varphi$ and a non-tilted perfect fluid with equation of state $p(t)=w\cdot\rho(t)$. This includes Bianchi types VIII and IX and the Kantowski-Sachs metric, in addition to the models in $\mathcal D (\text{I-VII}_h)$.	
	\item $\mathcal D (\text{I-VII}_h)$: the Bianchi types I-VII$_h$ subset of $\mathcal H$. A general representation based on a $1+1+2$ decomposition of Einstein’s field equation was given in \cite{Ben18}, and a representation with the orientation of the orthonormal frame fixed relative to matter anisotropies and geometrical anisotropies in \cite{Mik19}.  
	\item $\mathcal D^+(\text{III})$: Bianchi type III with an isotropy-violating gradient $\nabla_\mu \varphi$ that has a non-vanishing component along the tangent vectors of the $G_2$ subgroup of isometries. The defining conditions are given in \cite{Mik19}. Dimension: $d=7$.
	\item $\mathcal S^+(\text{III})$: LRS Bianchi type III with an isotropy-violating gradient $\nabla_\mu \varphi$ that is aligned with the LRS axis. This is the most general LRS subset of Bianchi type III with isotropy-violating $\nabla_\mu \varphi$. A representation is given by the dynamical system in section \ref{sub:DS:LRS}. Dimension: $d=3$.
	\item $\mathcal S^+_\text{SF}(\text{III})$: shear-free, Bianchi type III. A representation is the shear-free curve $\mathcal P_\text{SF}(\lambda)$ (with $\lambda>0$) in section \ref{ch:LRS}.  These are the only shear-free solutions with anisotropic spatial curvature within $\mathcal H$ \cite{Mik18}. Dimension: $d=1$.
	\item $\mathcal P_\text{SFA}$: the equilibrium point in $\mathcal S^+_\text{SF}(\text{III})$ with $\Omega=0$ and dynamically equivalent to Milne. Proved to be a global attractor within $\mathcal S^+(\text{III})$ for $w>-1/3$ in section \ref{ch:LRS}.   
\end{enumerate} 
The hierarchy of these sets is as follows:
\be
\mathcal P_\text{SFA} \in \mathcal S^+_\text{SF}(\text{III})  \subset \mathcal S^+(\text{III})   \subset  \mathcal D^+(\text{III})  \subset \mathcal D (\text{I-VII}_h)  \subset \mathcal H \,.
\ee

In terms of this classification, the part of Theorem \ref{thm:global} with $w>-1/3$ can be expressed in the following way:
\be
\lim_{\tau\rightarrow\infty} X = \mathcal P_\text{SFA} \text{\; for all \;} X\in S^+(\text{III})\,.
\ee
Of course, this implies the much weaker statement that $\mathcal P_\text{SFA}$ is a \emph{local} attractor within $\mathcal{S}^+(\text{III})$ \cite{Pereira:2016tmu}. But as our global result is restricted to LRS Bianchi type III, it is clear that a number of questions regarding the local stability remain open.  

\subsection{Neighborhood of the shear-free attractor}
First of all, one have the questions about what type of physical perturbations are possible in general relativity: 
\begin{itemize}
	\item \emph{Can $\mathcal P_\text{SFA}$ be perturbed into a non-LRS Bianchi model?}
	\newline Note that the LRS Bianchi type III is very special configuration within the space of Bianchi type III models. Specifically, $\mathcal S^+(\text{III}) \subset \mathcal D^+(\text{III})$, where $\mathcal S^+(\text{III})$ has dimension $d=3$ and $\mathcal D^+(\text{III})$ has dimension $d=7$. 
	\item \emph{Can $\mathcal P_\text{SFA}$ be perturbed into models beyond Bianchi type III?} \newline This question is motivated by the nontrivial way in which the Bianchi models are connected to each others. First, models of more special Bianchi type represent boundraries of models of more general Bianchi type. Second, the Bianchi classification is not a resolution to the equivalence problem. For instance LRS Bianchi type III models with $n^a_{\;\;a}\neq0$ are equivalent to LRS Bianchi type VIII models \cite{bok:MacCallum79}. 
\end{itemize}
Clearly, these questions must be answered before one can establish the stability with respect to all homogeneous perturbations.  

Lemma 5.1 in \cite{Mik19} provides a partial answer to these questions by stating that if $X\in \mathcal S^+_\text{SF}(\text{III})$ then $X+\delta X \in \mathcal D^+(\text{III})$ for all homogeneous perturbations $\delta X$. This statement was proved within the set $\mathcal D (\text{I-VII}_h)$ by using the constraint equations of the dynamical system explicitly. In order to prove a stronger version valid within $\mathcal H$ we need to consider Bianchi type VIII, Bianchi type IX and the Kantowski-Sachs metric. Since all state vectors $X\in \mathcal S^+_\text{SF}(\text{III})$ belong to LRS Bianchi type III it is geometrically disconnected from the Kantowski-Sachs metric (consider the metric (\ref{metric:LRSIII&KS})). As for Bianchi type VIII, which has LRS subsets equivalent to LRS Bianchi type III models with $n^a_{\;\;a}\neq 0$ \cite{bok:MacCallum79}, it is easier to consider the matter sector (the same argument applies to Bianchi type IX).  In \cite{Ben18} it was shown using the Klein-Gordon equation that $\nabla_\mu \varphi$ must be parallel to the normal congruence in Bianchi type VIII and IX model, i.e. the matter sector is perfect. But any state vector $X\in \mathcal S^+_\text{SF}(\text{III})$ breaks the isotropy of the matter sector by a finite non-zero spatial component of $\nabla_\mu \varphi$ and is thus disconnected from Bianchi type VIII and IX models. We have thus proved the following result:
\begin{lemma}
	Let $X\in \mathcal{S}^+_\text{SF}({\rm III})$. Then $X+\delta X \in \mathcal D^+({\rm III})$ for all spatially homogeneous perturbations $\delta X$.
	\label{lemma:neighborhood}
\end{lemma}  
\begin{proof}
Lemma 5.1 in \cite{Mik19} and the arguments above.
\end{proof}

Informally, within spatially homogeneous models ($\mathcal H$) the neighborhood of $\mathcal P_\text{SFA}$ belongs to the set $\mathcal D^+(\text{III})$. That is, homogeneous perturbations around $\mathcal P_\text{SFA}$ always fall into $\mathcal D^+(\text{III})$. The same is true for any point $\lambda > 0$ on the curve $\mathcal P_\text{SF}(\lambda)$. This is illustrated in Figure \ref{Fig2}. Note that the LRS subset $\mathcal S^+(\text{III})$ only accomodates a part of the perturbations around $\mathcal P_\text{SFA}$, or around any other point on $\mathcal P_\text{SF}(\lambda)$. Thus homogeneous perturbations generally break the LRS symmetry of $\mathcal P_\text{SF}(\lambda)$ and the answer to the first question above is \emph{yes}. But such perturbations are restricted to Bianchi type III and the answer to the second question is \emph{no}.
\begin{figure}[!ht]
	\centering	
	\begin{overpic}[width=0.95\textwidth,tics=10]{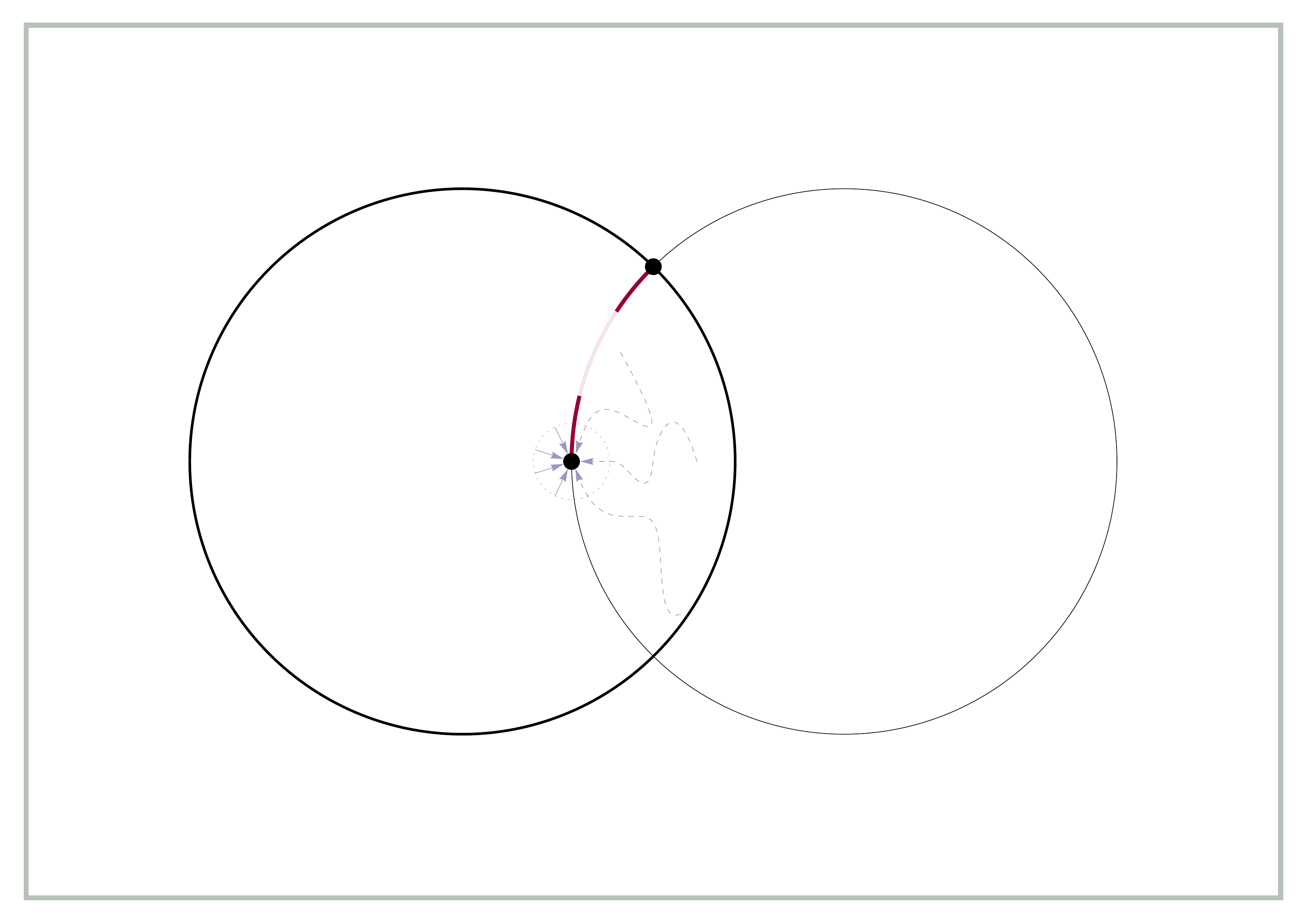}
		\put (37.5,31) {\footnotesize \textcolor{sfacolor}{$\displaystyle\mathcal P_\textrm{SFA}$ }}
		\put (48,53) {\footnotesize \textcolor{sfacolor}{$\displaystyle\mathcal P_\textrm{F}$ }}
		\put (8,62) {\footnotesize  \textcolor{boxcolor}{$\mathcal H$: Bianchi I-IX $\&$ K.S.}}
		\put (69,52) {\footnotesize\rotatebox{-21}{LRS} }
		\put (47.5,27) {\footnotesize\rotatebox{30}{$\mathcal{S}\displaystyle ^+$(III)} }
		\put (25,50) {\footnotesize\rotatebox{23}{$\mathcal{D}\displaystyle ^+$(III)} }
		\put (43.3,40.5) {\footnotesize\textcolor{sfcolor}{\rotatebox{67}{$\displaystyle\mathcal{P}_{\rm SF}(\lambda$)}} }
	\end{overpic}	
	\caption{Schematic illustration of the relation between the different invariant sets constituting the neighborhood of the equilibrium point $\mathcal P_{\rm SFA}$ (black dot). The circle marked 'LRS' represents all locally rotationally symmetric subsets of $\mathcal H$. Although $\mathcal P_\text{SFA}$ belongs to the LRS set $\mathcal S^+(\text{III})$, its neigborhood include non-LRS models as well.  The intersection of LRS models and $\mathcal D^+(\text{III})$ is $\mathcal{S}^+$(III) in which there is a monotonic function that proves the global stability of $\mathcal P_{\rm SFA}$ for $w>-1/3$ (Theorem \ref{thm:global}) illustrated here by curved, dashed arrows. Also, $\mathcal P_{\rm SFA}$ is locally stable w.r.t. all spatially homogeneous perturbations for $w>-1/3$ (Theorem \ref{thm:local}), as illustrated here by a circle of short arrows around the equilibrium point.  }
	\label{Fig2}
\end{figure}
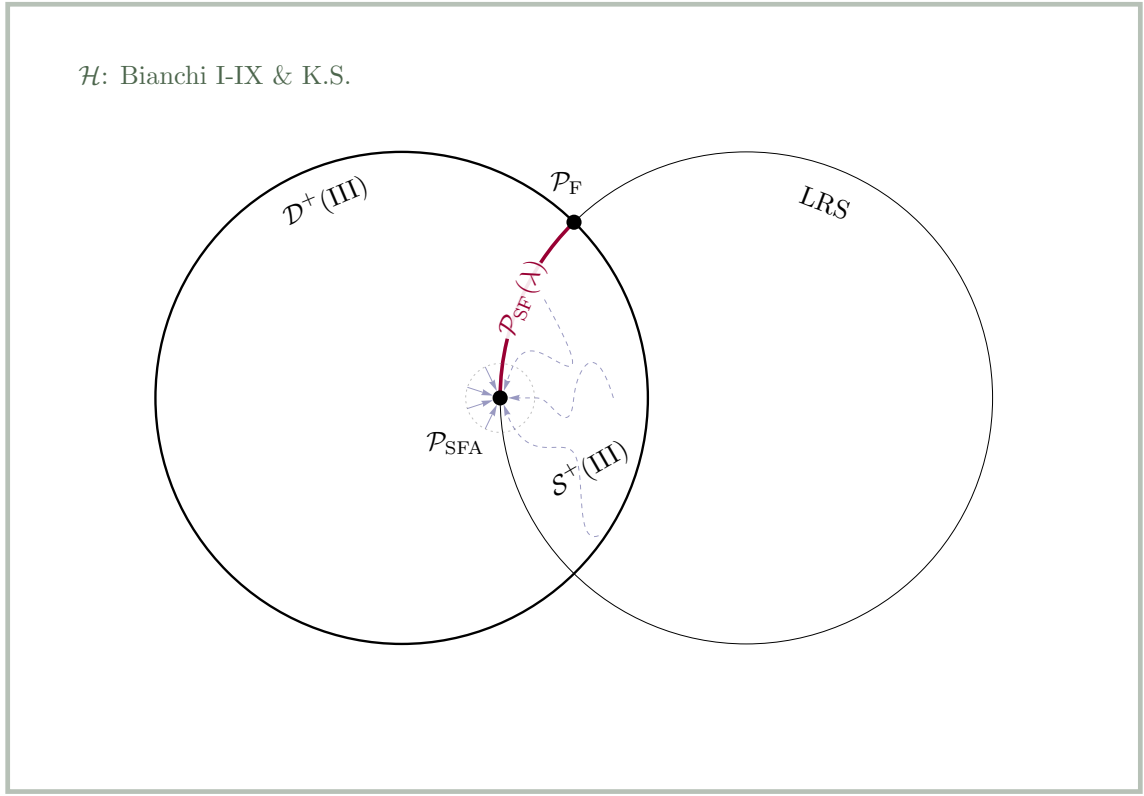

\subsection{Dynamical system \label{ch:dynamicalsystem}}
Above we established $\mathcal D^+(\text{III})$ as the neighborhood of $\mathcal P_\text{SFA}$  within spatially homogeneous cosmological models ($\mathcal H$). In order to establish the local stability of $\mathcal P_\text{SFA}$ with respect to all homogeneous perturbations, it is thus sufficient to consider the dynamical system associated with $\mathcal D^+(\text{III})$. Our starting point is the dynamical system associated with the set $\mathcal D (\text{I-VII}_h)$ given in \cite{Mik19}. This is an autonomous system of first order equations for 11 variables subject to 5 non-linear constraint equations. By inserting the defining conditions $A=\sqrt{3}N_\times>0$ and $N_+ = \sqrt{3} N_-$ we obtain the dynamical system below, which is associated with the invariant set $\mathcal D^+(\text{III})$. It is important to note that, unlike the LRS subsystem (\ref{DS:S})-(\ref{DS:p}) this system generally breaks all the axiomatic conditions (\ref{ansatz:1})-(\ref{ansatz:3}) of the matter ansatz. Specifically, we note that $p_\varphi / \rho_\varphi \in [-\tfrac{1}{3}, 1]$ is a time-dependent quantity and that the energy flux $q^\mu$ is generally non-zero\footnote{$q^\mu=0$ corresponds to the special cases where $\nabla_\mu \varphi$ is parallel to or orthogonal to the normal congruence.}.

\paragraph{Equations of motion:}
\begin{align} 
	\Theta' &= \Theta (q-2) -2\sqrt{3} N_\times V_1 \,, \label{new:T} \\ 
	V_1' &= V_1(q+2\Sigma_+)-2\sqrt{3} \Sigma_3 V_3 \,, \label{new:V1} \\
	V_3' &= V_3(q-\Sigma_+ + \sqrt{3} \Sigma_- ) \,, \label{new:V3} \\
	\Sigma_+' &= \Sigma_+ (q-2)+3\Sigma_3^2 - 2 (N_-^2+N_\times^2) + V_3^2 -2V_1^2 \,, \label{new:Sp} \\
	\Sigma_-' &= \Sigma_- (q-2)-\sqrt{3}(V_3^2 - 2\Sigma_\times^2 + \Sigma_3^2 +2N_-^2- 2 N_\times^2)   
	\,, \label{new:Sm} \\
	\Sigma_\times' &= \Sigma_\times (q-2-2\sqrt{3} \Sigma_-) - 4\sqrt{3} N_\times N_-  \,, \label{new:Sc} \\
	\Sigma_3' &= \Sigma_3 (q-2-3\Sigma_++\sqrt{3}\Sigma_-)+2\sqrt{3}V_1V_3 \,, \label{new:S3}  \\
	N_-' &= N_- (q+2\Sigma_+) + 2\sqrt{3}(\Sigma_\times N_\times + \Sigma_- N_-)  \,, \\
	N_\times' &= N_\times (q+2\Sigma_+) \,, \label{new:Nx} 
\end{align}
where 
\begin{align}
q &= \frac{1}{2}(1+3w)\Omega + 2\left(\Theta^2+\Sigma_+^2 +\Sigma_-^2 +\Sigma_\times^2 + \Sigma_3^2\right)\,, \\
\Omega &= 1-\Theta^2-V_1^2-V_3^2-\Sigma_+^2 -\Sigma_-^2 -\Sigma_\times^2 - \Sigma_3^2 - N_-^2 - 4 N_\times^2 \,.
\label{new:def:sigma}
\end{align}
\paragraph{Constraints:}
\begin{align}
\sqrt{3} N_\times \Sigma_+ + \Sigma_- N_\times - \Sigma_\times N_- - \Theta V_1 &= 0 \, ,\label{C1} \\
N_\times \Sigma_3 + \Theta V_3 &= 0 \,, \label{C2} \\
\Theta^2 + V_1^2 + V_3^2  + \Sigma_+^2 +\Sigma_-^2 +\Sigma_\times^2 + \Sigma_3^2  + N_-^2 + 4N_\times^2 &\le 1 \,, \label{inequality:1} \\
N_\times &> 0 \,, \label{inequality:2} \\
V_3 &> 0 \,. \label{inequality:3}
\end{align}
Here the orthonormal frame is fixed uniquely relative to geometrical anisotropies and matter anisotropies. Since the gauge is fixed completely all variables represent physical degrees of freedom and the dimension of the dynamical system is the number of variables minus the number of constraint equations: $d=9-2=7$. Among these variables, $\{\Theta, V_1, V_3\}$ are associated with $X_\mu$ and describe the gradient of the scalar field $\varphi$, $\{\Sigma_+, \Sigma_-, \Sigma_\times, \Sigma_3\}$ are associated with $\Sigma_{ab}$ and describe the shear seen by an observer comoving with the perfect fluid (normal congruence), and $\{N_-, N_\times\}$ are associated with $N_{ab}$ and $A_a$ and describe the spatial curvature. As for the inequalities, it should be commented that (\ref{inequality:1}) corresponds to the positive energy condition $\Omega\ge0$ and that (\ref{inequality:2}) ensures Bianchi type III ($A=\sqrt{3}N_\times\rightarrow 0$ is the Bianchi type II boundrary). Discrete symmetries of the dynamical system have been employed to choose the positive sign of $V_3$ without loss of generality. See \cite{Mik19} for further details.

\paragraph{LRS subsystem} Next let us identify the LRS subsystem $\mathcal S^+(\text{III})$ and connect it to the metric/coordinate approach used in section \ref{ch:LRS}. It is obtained by imposing the following conditions on $\mathcal D^+(\text{III})$:
\be
\mathcal S^+(\text{III}): \quad \Theta=V_1=\Sigma_\times = \Sigma_3  = N_- = 0\,, \quad \Sigma_-=-\sqrt{3} \Sigma_+ \,. \;
\label{SIII}
\ee
In fact, this is the only representation of the LRS subsystem with the employed gauge fixing prescription \cite{Mik19}. Using the evolution equations above, it is easy to check that the conditions (\ref{SIII}) are preserved in time, that is, $\mathcal S^+(\text{III})$ forms an invariant subset of $\mathcal D^+(\text{III})$. Note that the constraint equations (\ref{C1}) and (\ref{C2}) are satisfied trivially so the remaning variables $(V_3, \Sigma_+, N_\times)$ are independent and the dimension of $\mathcal S^+(\text{III})$ is $d=3$. Also note that the normalized shear tensor takes the form $\Sigma_{ab}=\text{diag}(-2\Sigma_+, -2\Sigma_+, 4\Sigma_+)$ and that the LRS axis is identified with the third basis vector $\mathbf e_3$. The connection to the metric approach employed in section \ref{ch:LRS} is given by 
\be
\Omega_\varphi=V_3^2\,, \quad \Sigma=-2\Sigma_+ \,, \quad \Omega_k=4N_\times^2 \,.
\label{topdown}
\ee
Given this identification the  dynamical system (\ref{DS:S})-(\ref{DS:p}) is reproduced ``top-down'' from the evolution equations above. 

\paragraph{Shear-free subsystem}
The shear-free Bianchi type III model is an invariant subset of $\mathcal D^+(\text{III})$ (and $\mathcal S^+(\text{III})$) defined by the following conditions:
\be
\mathcal S_\text{SF}^+(\text{III}): \quad \Theta=V_1 = \Sigma_-= \Sigma_+ = \Sigma_\times = \Sigma_3 = N_- = 0\,, \quad V_3 = \sqrt{2} N_\times > 0 \,. \; 
\label{SIII:SF}
\ee
The dimension of the set is $d=1$.

\subsection{Local stability analysis \label{ch:localstability}}
Let us reference the coordinates of $\mathcal D^+(\text{III})$ with respect to the basis
\be
\left\{ \Theta, V_1, V_3, \Sigma_+, \Sigma_-, \Sigma_\times, \Sigma_3, N_-, N_\times \right\} \,.
\label{basis}
\ee
Each state vector $X\in\mathcal D^+(\text{III})$ is thus seen as a point in an extended state space $\mathbb R^9$. The physical part of $\mathbb R^9$, that corresponds to the state space $\mathcal D^+(\text{III})$, is the part that satisfies the constraint equations (\ref{C1})-(\ref{C2}) and the inequalities (\ref{inequality:1})-(\ref{inequality:3}).

Shear-free solutions lay along the following curve in $\mathbb R^9$:
\be
\mathcal{P}_\text{SF}(\lambda) = \lambda \left(0, 0, \frac{1}{\sqrt{3}}, 0, 0, 0, 0, 0, \frac{1}{\sqrt{6}} \right) \,, \quad \lambda\in[0,1]\,,
\label{eq:PSFA}
\ee
where $\lambda$ again evolves according to (\ref{ev:lambda}). The set $\mathcal S_\text{SF}^+(\text{III})$ corresponds to the half-open interval $\lambda\in (0,1]$. As with the parametrization (\ref{curve:3D}) the flat Friedmann solution and the shear-free attractor correspond to the end points:
\be
\mathcal P_\text{F} = \mathcal P_\text{SF}(0)\,,\qquad \mathcal P_\text{SFA} = \mathcal P_\text{SF}(1) \,.
\ee
These are the only shear-free equilibrium points for $w\neq-1/3$. For $w=-1/3$ $\mathcal{P}_\text{SF}(\lambda)$ is a one-parameter family of equilibrium points. 

First we keep $w$ free and consider arbitrary perturbations around $\mathcal P_\text{SFA}$:
\be
\mathrm{X} = \mathcal{P}_\text{SF}(1) + \delta\mathrm{X}\,.
\ee
Constraint equations (\ref{C1}) and (\ref{C2}) are used to eliminate the variables $\Sigma_+$ and $\Sigma_3$. The linearized equations can then be written $\mathrm{\delta X}' = \mathrm{M} \, \mathrm{\delta X}$, where $\mathrm{M}$ is a 7 by 7 matrix with eigenvalues
\be
-1 + \sqrt{3} i, \; -1 -\sqrt{3}i, 
-1 - 3 w \,,
\nonumber
\ee
where the first two eigenvalues have a 3-fold degeneracy each. For $w>-1/3$ the real part of all the eigenvalues are negative. Thus $\mathcal{P}_\text{SFA}$ is stable and an attractor. For $w<-1/3$ the last eigenvalue is positive, so in this case $\mathcal{P}_\text{SFA}$ is a saddle. The unstable direction is along $\mathcal{P}_\text{SF}(\lambda)$ towards the flat FLRW solution. This is consistent with the no-hair theorem proved for this class of models in \cite{Ben18}. 

Next we consider the particular case $w=-1/3$. Now, $\mathcal{P}_\text{SF}(\lambda)$ is a one-parameter family of equilibrium points. At each point we consider arbitrary perturbations:
\be
\mathrm{X} = \mathcal{P}_\text{SF}(\lambda) + \delta\mathrm{X} \,.
\ee
Again there are only three distinct eigenvalues:
\be
-1 \pm \sqrt{1-4\lambda^2}, \quad 0 \,.
\nonumber
\ee 
There is only one zero eigenvalue and three copies of each of the other two. Thus at each point $\lambda\in(0, 1]$ there is one zero eigenvalue and 6 eigenvalues with negative real part. The zero eigenvalue is associated with the eigenvector $\tfrac{d}{d\lambda} \mathcal P_\textrm{SF}$ and thus merely reflects perturbations along the fix-point curve itself. We conclude that, as a whole, the set of fix points given by the curve $\mathcal P_\text{SF}(\lambda)$ is stable.

We have thus showed that the equilibrium sets of $\mathcal S^+_\text{SF}$ are stable with respect to all homogeneous perturbations. Each equilibrium point corresponds to a unique \emph{transitively self-similar cosmological model}, informally, an expanding cosmological model whose physical states at different times are similar up to a transformation of length scale (that preserves the expansion normalized variables).\footnote{See section 5.3.3 of \cite{bok:EllisWainwright} for technical definition.} Since $\mathcal S^+_\text{SF}$ are the only shear-free solutions with anisotropic spatial curvature, within our considered class of spatially homogeneous models ($\mathcal H$), the local analysis above in combination with Lemma \ref{lemma:neighborhood} gives the following result:
\newpage
\begin{theorem}
All equilibrium sets in $\mathcal H$ that correspond to a transitively self-similar eFLRW model (definition \ref{def:eFLRW}) are stable with respect to all homogeneous perturbations for $w\ge -1/3$. Specifically:
 \begin{enumerate}
 	\item[a)] Within spatially homogeneous models $\mathcal H$ the only shear-free transitively self-similar eFLRW models are those that correspond to $\mathcal P_\text{SFA}$ for all $w\in[-1,1]$ and $\mathcal{P}_\text{SF}(\lambda)$ with $\lambda \in (0,1]$ for $w=-1/3$. 
 	\item[b)] The equilibrium point $\mathcal P_\text{SFA}$ is stable with respect to all homogeneous perturbations for $w>-1/3$.
 	\item[c)] The equilibrium set $\mathcal{P}_\text{SF}(\lambda)$ ($w=-1/3$) as a whole, which includes $\mathcal P_\text{SFA}=\mathcal P_\text{SF}(1)$, is stable with respect to all homogeneous perturbations.
 \end{enumerate}
	\label{thm:local}
\end{theorem}
\begin{proof}
\begin{enumerate}
\emph{a):} Theorem 1 in \cite{Mik19} proves that $\mathcal S^+_\text{SF}$ represents all shear-free ($\sigma_{\mu\nu}=0$) solutions with anisotropic spatial curvature (${^3}S_{\mu\nu}\neq0$) within $\mathcal D (\text{I-VII}_h)$. To prove that this result is valid within $\mathcal H$, use the same argument as in Lemma \ref{lemma:neighborhood} to prove the absence of such solutions in Bianchi type VIII and Bianchi type IX and refer to section 5.6 of \cite{Mik18} for the Kantowski-Sachs metric.  
\emph{b) and c):} Use Lemma \ref{lemma:neighborhood} in combination with the local stability analysis above.
\end{enumerate}
\end{proof}

\section{Conclusion}\label{sec:Conclusion}

In this paper we first reviewed how one can construct shear-free cosmological models with underlying anisotropic spatial geometry, or \emph{extended FLRW models}, by introducing an imperfect fluid. We formulated a \emph{matter ansatz} with three conditions (\ref{ansatz:1})-(\ref{ansatz:3}) on the energy-momentum tensor that are necessary for such solutions to exist and be dynamically equivalent to a reference FLRW model. All FLRW models can be considered as special, or ``tuned'', in the sense that the shear tensor is exactly zero. But the axiomatic conditions (\ref{ansatz:1})-(\ref{ansatz:3}), that underlie our definition \ref{def:eFLRW} of an extended FLRW model, represent a further tuning of configuration space. The main scope of this paper was to investigate the naturalness of this ansatz in a dynamical system context by a thorough stability analysis. Our investigation is restricted to extended FLRW models where the matter ansatz is realized by a massless scalar field $\varphi$ that respects the positive energy condition $\rho_\varphi>0$. The advantage of using a concrete physical model for the required imperfect fluid is that no ad-hock assumptions were necessary on the energy-momentum tensor. A massless scalar field,  with a spatially homogeneous gradient $\nabla_\mu \varphi$ that breaks isotropy, is capable of satisfying the three conditions (\ref{ansatz:1})-(\ref{ansatz:3}), yet rich enough to break all of them. 

Our choice of matter model picks out the LRS Bianchi type III metric as the unique spatially homogeneous spacetime that can accomodate an extended FLRW model. This solution is represented by the curve $\mathcal{P}_\text{SF}(\lambda), \lambda\in [0,1]$, that connects the shear-free attractor $\mathcal P_\text{SFA}$ ($\lambda=1$) with the flat FLRW solution $\mathcal P_\text{F}$ ($\lambda = 0$). Our main results for the stability of the equilibrium points of $\mathcal{P}_\text{SF}(\lambda)$ are given by Theorem \ref{thm:global} and Theorem \ref{thm:local}, whose content is illustrated in Figure \ref{Fig2}. Theorem \ref{thm:global} establishes the stability of the solution within LRS Bianchi type III models, that is, within the invariant set $\mathcal S^+(\text{III})$. In this simple model we where able to find a monotonic function and prove the stability globally. But one should note that in $\mathcal S^+(\text{III})$, conditions (\ref{ansatz:2}) and (\ref{ansatz:3}) are identities that result directly from the underlying theory (Einstein equations sourced by a massless scalar field). Theorem \ref{thm:local} establishes the stability with respect to all homogeneous perturbations compatible with our matter model. Such perturbations are generally non-LRS and generally break all conditions (\ref{ansatz:1})-(\ref{ansatz:3}) of the matter ansatz. Remarkably, we found the equilibrium points of $\mathcal{P}_\text{SF}(\lambda)$ to be stable with respect to all homogeneous perturbations. We have thus shown that geometrical anisotropies and matter anisotropies tend to rearrange dynamically to a configuration where the ansatz for the energy-momentum tensor is satisfied. Extended FLRW models requires a matter sector fine-tuned with respect to geometrical anisotropies, yes, but this is indeed a dynamically natural configuration. 

Finally, we would like to add some comments about the viability of eFLRW models in a broader cosmological context. In particular, since they share the same background dynamics as FLRW models, to what extent can these models be distinguished on the basis of recent cosmological observations? Owing to the presence of anisotropic curvature in eFLRW models, observational signatures are expected to appear either at the level of the background geometry, or in the dynamics of linear cosmological perturbations. At the background level, it has been shown that the anisotropic curvature will alter both the luminosity and angular diameter distances, which in principle allows one to test these models through the observed angular distribution of type Ia supernovae~\cite{main:Koivisto11,Menezes:2012kc}. At the perturbative level, eFLRW models predict
off-diagonal correlations in the temperature spectrum of CMB~\cite{Pereira:2015pxa}, as well as different dynamics for the polarization of tensor perturbations~\cite{pert:Pereira17}. However, since these signatures are stronger at the lowest CMB multipoles, where cosmic variance is high, it seems unlikely that extended FLRW models can be ruled out on the basis of current CMB data. In either cases, it is expected that forthcoming cosmological observations will have enough constraining power to give nature's verdict about this class of models. In the meantime, they remain as viable phenomenological candidates.

\paragraph{Acknowledgements}
We thank Hans A. Winther for useful comments on the manuscript. TSP thanks Brazilian Funding Agencies CNPq (311527/2018-3) and Fundação Araucária (CP/PBA-2016) for their financial support.

\appendix
\section{Photon propagation and isotropy of the CMB}
\label{sec:PhotonProp}
It is well-known that the background isotropy of the CMB remains intact in cosmological models that admit a conformal Killing vector that is parallel to the fundamental congruence \cite{Ehlers:1966ad,Clarkson:1999yj,Obukhov:2000jf}. This is the case for eFLRW models with the metric (\ref{line-element2}), where the conformal factor is $a^2(t)$. Nevertheless, it is instructive to see how the details turn out by a worked example. Here we explicitly calculate null geodesics in eFLRW models of LRS Bianchi type III and the Kantowski-Sachs metric and show that the isotropy of the background CMB remains intact. The interpretation is that the true source of anisotropy of the CMB in general Bianchi models \cite{Barrow:1985tda, Jaffe:2005} is the rate of shear, not other anisotropies such as that of three-dimensional curvature or vorticity; see \cite{int:Obukhov1992} for a clear discussion. 

We start from the following line element that cover LRS spacetimes of Bianchi type III ($k\,<\,0$) as well as the Kantowski-Sachs metric ($k\,>\,0$): 

\be
ds^2=-dt^2+e^{2\alpha(t)} \left( dz^2 + \frac{d\rho^2}{1-k\rho^2} + \rho^2 d\phi^2  \right)\,.
\label{metric:LRSIII&KS}
\ee
The spatial sections of Bianchi type III is $\mathcal{R}_1 \times \mathcal{H}_2$, whereas Kantowski-Sachs has spatial sections $\mathcal{R}_1 \times \mathcal{S}_2$. We take the two-dimensional slices $\mathcal{H}_2$ and $\mathcal{S}_2$ to be parametrized by the radial coordinate $\rho$ and the angular coordinate $\phi$. Orthogonal to these is the LRS axis parametrized by the coordinate $z$. Due to the homogeneity of the spatial sections we can without loss of generality put the observer at the origin of the coordinate system, i.e., $z=\rho=0$.  Take $\lambda$ to be an affine parameter parametrizing the path of the photons towards the observer at the origin. Then, due to the plane symmetry of these spacetimes one must have $d\phi/d\lambda=0$. Next we turn our attention to the geodesic equation, which reads
\be
\label{photonconstr}
\frac{d^2x^\mu}{d\lambda^2} + \Gamma^\mu_{\alpha\beta} \frac{dx^\alpha}{d\lambda}\frac{dx^\beta}{d\lambda} = 0\,.
\ee
Here $\{x^\mu\}\,=\,\{t,z,\rho,\phi\}$ is the set of coordinates. Invoking $d\phi/d\lambda=0$ in this equation, and calculating as usual, one ends up with the equation
\begin{align}
&\frac{dU}{d\lambda} + e^{-2\alpha}H(P^2+Q^2) = 0\,, \label{eq:U}
\end{align} 
where $U\equiv dt/d\lambda$, alongside two constants of motion:
\begin{align}
P=e^{2\alpha}\frac{dz}{d\lambda} \quad\quad\textrm{and}\quad\quad Q=\frac{e^{2\alpha}}{\sqrt{1-k\rho^2}} \frac{d\rho}{d\lambda} \,.
\end{align}
To connect the above with standard notation: note that the scale factor $a(t)$ is defined such that $a(t)=e^{\alpha(t)}$. Hence $H(t)=\dot{a}/a=\dot{\alpha}$. 
An additional constraint results from the fact that photons follow null-geodesics. Specifically,
\be
\label{eq:constrPhoton}
g_{\mu\nu} \frac{dx^\mu}{d\lambda} \frac{dx^\nu}{d\lambda} = 0 \quad\quad\rightarrow\quad\quad U^2=e^{-2\alpha}(P^2+Q^2)\,.
\ee

\paragraph{Bending of the photon path. }Next we want to calculate the bending of the null-geodesic relative to the intrinsic LRS-axis (the $z$-axis), as a function of $\lambda$.  Take $v^\mu$ to be the components of the unit vector parallel to this axis. Furthermore: take $w^\mu$ to be the components of the unit vector parallel to the \emph{spatial} direction of the four velocity of the photon. It follows that
\be
v^\mu = e^{-\alpha} \left(0,1,0,0\right)\quad\quad\textrm{and}\quad\quad w^\mu = U^{-1}\left( 0,\frac{dz}{d\lambda}, \frac{d\rho}{d\lambda},0\right).
\ee
The angle $\theta$ between the vectors $\mathbf{v}$ and $\mathbf{w}$ is now found through the inner product. In particular:
\be
\cos \theta = g_{\mu\nu} v^\mu w^\mu = \frac{P}{\sqrt{P^2+Q^2}}\,,
\label{eq:deftheta}
\ee
where in the last step we have used (\ref{eq:constrPhoton}). Now, since $P$ and $Q$ are constants of motion, it follows that the photon propagates with a fixed angle relative to the LRS axis.

\paragraph{Redshift of the photons.} Consider two photons emitted in the same direction with a small time separation $\tau(\lambda)$. The time coordinates of the two photons is denoted as $t_1(\lambda)=t(\lambda)$ and $t_2(\lambda)=t(\lambda)+\tau(\lambda)$.  The redshift $Z(\lambda)$ of a photon emitted at a time $t(\lambda)$ is implicitely defined through
\be
1+Z(\lambda) = \frac{\tau_0}{\tau(\lambda)}\,,
\label{eq:defz}
\ee  
where $\tau_0=\tau(\lambda_0)$ is the time separation today. Differentiation with respect to the e-fold time parameter $\alpha(\lambda)$ yields
\be
\frac{dZ}{d\alpha} = -\frac{1}{\tau} \frac{d\tau}{d\lambda} \left( \frac{1+Z}{UH} \right)\,. 
\label{eq:dzdalpha}
\ee
To calculate the factor $d\tau/d\lambda$ on the right-hand side we consider the four-velocity identity (\ref{eq:constrPhoton}) for each of the two photons:
\begin{align}
U^2(t_1) &=e^{-2\alpha(t_1)}(P^2+Q^2)_1\,, \label{nr1} \\
U^2(t_2) &=e^{-2\alpha(t_2)}(P^2+Q^2)_2\,. \label{nr2}   
\end{align}
$P^2+Q^2$ is a constant of motion, but the numerical value is different for the two photons and we have indicated this by the subscripts.  Using (\ref{eq:deftheta}) we get 
\begin{align}
(P^2+Q^2)_1 = \cos^{-2}\!\theta_1 \frac{dx}{d\lambda} e^{4\alpha(t_1)}\quad\quad\textrm{and}\quad\quad (P^2+Q^2)_2 = \cos^{-2}\!\theta_2 \frac{dx}{d\lambda} e^{4\alpha(t_2)}\,.
\end{align}
Since we are considering two photons emitted (and detected) in the same direction, we have $\theta_1=\theta_2$ (recall that $\theta$ is a constant of motion). Expanding to first order in $\tau$ we get
\be
(P^2+Q^2)_2 = (1+4H\tau)(P^2+Q^2)_1\,. \label{nr3}
\ee  
Next we expand (\ref{nr2}) to first order in $\tau$
\be
U^2(t)+2U(t)\frac{d\tau}{d\lambda} =  e^{-2\alpha(t)} \left( 1 -2H\tau   \right) (P^2+Q^2)_2\,.
\label{nr5}
\ee
Inserting (\ref{nr1}) and (\ref{nr3}) and neglecting the second order term in $\tau$ we find the sought after expression:
\be 
\frac{1}{\tau}\frac{d\tau}{d\lambda} = HU\,. 
\ee
Now the differential equation for the redshift (\ref{eq:dzdalpha}) simplifies to
\be
\frac{dZ}{d\alpha} = -(1+Z)\,,
\ee
and it is clear that the redshift of a photon is independent of its direction of propagation. Integration with $Z(\lambda_0)=0$ gives: 
\be
1+Z = \frac{a_0}{a}\,,
\ee 
where $a=e^{\alpha}$ is the scale factor when the photon is emitted and $a_0$ the scale factor today. Thus we see how the background CMB is not changed relative to FLRW in these spacetimes.

\appendix

\bibliographystyle{JHEP}
\bibliography{refs}

\providecommand{\href}[2]{#2}\begingroup\raggedright\begin{thebibliography}{10}

\bibitem{Smoot:1992td}
{\scshape COBE} collaboration, G.~F. Smoot et~al., \emph{{Structure in the COBE
  differential microwave radiometer first year maps}},
  \href{https://doi.org/10.1086/186504}{\emph{Astrophys. J.} {\bfseries 396}
  (1992) L1--L5}.

\bibitem{Komatsu:2008hk}
{\scshape WMAP} collaboration, E.~Komatsu et~al., \emph{{Five-Year Wilkinson
  Microwave Anisotropy Probe (WMAP) Observations: Cosmological
  Interpretation}},
  \href{https://doi.org/10.1088/0067-0049/180/2/330}{\emph{Astrophys. J.
  Suppl.} {\bfseries 180} (2009) 330--376},
  [\href{https://arxiv.org/abs/0803.0547}{{\ttfamily 0803.0547}}].

\bibitem{Komatsu:2010fb}
{\scshape WMAP} collaboration, E.~Komatsu et~al., \emph{{Seven-Year Wilkinson
  Microwave Anisotropy Probe (WMAP) Observations: Cosmological
  Interpretation}},
  \href{https://doi.org/10.1088/0067-0049/192/2/18}{\emph{Astrophys. J. Suppl.}
  {\bfseries 192} (2011) 18},
  [\href{https://arxiv.org/abs/1001.4538}{{\ttfamily 1001.4538}}].

\bibitem{Hinshaw:2012aka}
{\scshape WMAP} collaboration, G.~Hinshaw et~al., \emph{{Nine-Year Wilkinson
  Microwave Anisotropy Probe (WMAP) Observations: Cosmological Parameter
  Results}}, \href{https://doi.org/10.1088/0067-0049/208/2/19}{\emph{Astrophys.
  J. Suppl.} {\bfseries 208} (2013) 19},
  [\href{https://arxiv.org/abs/1212.5226}{{\ttfamily 1212.5226}}].

\bibitem{Ade:2013zuv}
{\scshape Planck} collaboration, P.~A.~R. Ade et~al., \emph{{Planck 2013
  results. XVI. Cosmological parameters}},
  \href{https://doi.org/10.1051/0004-6361/201321591}{\emph{Astron. Astrophys.}
  {\bfseries 571} (2014) A16},
  [\href{https://arxiv.org/abs/1303.5076}{{\ttfamily 1303.5076}}].

\bibitem{Ade:2015xua}
{\scshape Planck} collaboration, P.~A.~R. Ade et~al., \emph{{Planck 2015
  results. XIII. Cosmological parameters}},
  \href{https://doi.org/10.1051/0004-6361/201525830}{\emph{Astron. Astrophys.}
  {\bfseries 594} (2016) A13},
  [\href{https://arxiv.org/abs/1502.01589}{{\ttfamily 1502.01589}}].

\bibitem{planck18}
{\scshape Planck} collaboration, N.~Aghanim et~al., \emph{{Planck 2018 results.
  VI. Cosmological parameters}},
  \href{https://arxiv.org/abs/1807.06209}{{\ttfamily 1807.06209}}.

\bibitem{PlanckIsotropy:2019}
{\scshape Planck} collaboration, Y.~Akrami et~al., \emph{{Planck 2018 results.
  VII. Isotropy and Statistics of the CMB}},
  \href{https://arxiv.org/abs/1906.02552}{{\ttfamily 1906.02552}}.

\bibitem{DiValentino:2019}
E.~Di~Valentino, A.~Melchiorri and J.~Silk, \emph{{Planck evidence for a closed
  Universe and a possible crisis for cosmology}},
  \href{https://doi.org/10.1038/s41550-019-0906-9}{\emph{Nat. Astron.} (2019)
  }, [\href{https://arxiv.org/abs/1911.02087}{{\ttfamily 1911.02087}}].

\bibitem{Copi:2010na}
C.~J. Copi, D.~Huterer, D.~J. Schwarz and G.~D. Starkman, \emph{{Large angle
  anomalies in the CMB}}, \href{https://doi.org/10.1155/2010/847541}{\emph{Adv.
  Astron.} {\bfseries 2010} (2010) 847541},
  [\href{https://arxiv.org/abs/1004.5602}{{\ttfamily 1004.5602}}].

\bibitem{Schwarz:2015cma}
D.~J. Schwarz, C.~J. Copi, D.~Huterer and G.~D. Starkman, \emph{{CMB Anomalies
  after Planck}},
  \href{https://doi.org/10.1088/0264-9381/33/18/184001}{\emph{Class. Quant.
  Grav.} {\bfseries 33} (2016) 184001},
  [\href{https://arxiv.org/abs/1510.07929}{{\ttfamily 1510.07929}}].

\bibitem{Bull:2015stt}
P.~Bull et~al., \emph{{Beyond $\Lambda$CDM: Problems, solutions, and the road
  ahead}}, \href{https://doi.org/10.1016/j.dark.2016.02.001}{\emph{Phys. Dark
  Univ.} {\bfseries 12} (2016) 56--99},
  [\href{https://arxiv.org/abs/1512.05356}{{\ttfamily 1512.05356}}].

\bibitem{bok:EllisWainwright}
J.~Wainwright and G.~F.~R. Ellis, \emph{{Dynamical Systems in Cosmology}}.
\newblock Cambridge University Press, 2005.

\bibitem{bok:bondi68}
H.~Bondi, \emph{Cosmology}.
\newblock Cambridge Monographs on Physics. Cambridge University Press, 1968.

\bibitem{Uzan:2016wji}
J.-P. Uzan, \emph{{The big-bang theory: construction, evolution and status}},
  \href{https://arxiv.org/abs/1606.06112}{{\ttfamily 1606.06112}}.

\bibitem{Collins:1973lda}
C.~B. Collins and S.~W. Hawking, \emph{{The rotation and distortion of the
  Universe}}, {\emph{Mon. Not. Roy. Astron. Soc.} {\bfseries 162} (1973)
  307--320}.

\bibitem{Bunn:1996ut}
E.~F. Bunn, P.~Ferreira and J.~Silk, \emph{{How anisotropic is our universe?}},
  \href{https://doi.org/10.1103/PhysRevLett.77.2883}{\emph{Phys. Rev. Lett.}
  {\bfseries 77} (1996) 2883--2886},
  [\href{https://arxiv.org/abs/astro-ph/9605123}{{\ttfamily
  astro-ph/9605123}}].

\bibitem{Koivisto:2008ig}
T.~Koivisto and D.~F. Mota, \emph{{Anisotropic Dark Energy: Dynamics of
  Background and Perturbations}},
  \href{https://doi.org/10.1088/1475-7516/2008/06/018}{\emph{JCAP} {\bfseries
  0806} (2008) 018}, [\href{https://arxiv.org/abs/0801.3676}{{\ttfamily
  0801.3676}}].

\bibitem{Hsu-Wainwright:1986}
L.~Hsu and J.~Wainwright, \emph{{Self similar spatially homogeneous
  cosmologies: Orthogonal perfect fluid and vacuum solutions}},
  \href{https://doi.org/10.1088/0264-9381/3/6/011}{\emph{Class. Quant. Grav.}
  {\bfseries 3} (1986) 1105--1124}.

\bibitem{Wainwright-Hsu:1989}
J.~Wainwright and L.~Hsu, \emph{{A dynamical systems approach to Bianchi
  cosmologies: Orthogonal models of class A}},
  \href{https://doi.org/10.1088/0264-9381/6/10/011}{\emph{Class. Quant. Grav.}
  {\bfseries 6} (1989) 1409--1431}.

\bibitem{Hewitt-Wainwright:1993}
C.~G. Hewitt and J.~Wainwright, \emph{{A Dynamical systems approach to Bianchi
  cosmologies: Orthogonal models of class B}},
  \href{https://doi.org/10.1088/0264-9381/10/1/012}{\emph{Class. Quant. Grav.}
  {\bfseries 10} (1993) 99--124}.

\bibitem{hawking1969rotation}
S.~Hawking, \emph{On the rotation of the universe}, {\emph{Monthly Notices of
  the Royal Astronomical Society} {\bfseries 142} (1969) 129--141}.

\bibitem{Barrow:1985tda}
J.~D. Barrow, R.~Juszkiewicz and D.~H. Sonoda, \emph{{Universal rotation - How
  large can it be?}}, {\emph{Mon. Not. Roy. Astron. Soc.} {\bfseries 213}
  (1985) 917--943}.

\bibitem{int:Obukhov1992}
Y.~N. Obukhov, \emph{Rotation in cosmology},
  \href{https://doi.org/10.1007/BF00756780}{\emph{General Relativity and
  Gravitation} {\bfseries 24} (Feb, 1992) 121--128}.

\bibitem{Obukhov:2000jf}
Y.~N. Obukhov, \emph{{On physical foundations and observational effects of
  cosmic rotation}},  in \emph{{Published in Colloquium on Cosmic Rotation:
  Proceedings. Edited by M. Scherfner, T. Chrobok and M. Shefaat (Wissenschaft
  und Technik Verlag: Berlin, 2000) pp. 23-96}}, pp.~23--96, 2000,
  \href{https://arxiv.org/abs/astro-ph/0008106}{{\ttfamily astro-ph/0008106}}.

\bibitem{Maccallum:1970eg}
M.~A.~H. Maccallum, J.~M. Stewart and B.~G. Schmidt, \emph{{Anisotropic
  stresses in homogeneous cosmologies}},
  \href{https://doi.org/10.1007/BF01646029}{\emph{Commun. Math. Phys.}
  {\bfseries 17} (1970) 343--347}.

\bibitem{Barrow:1997sy}
J.~D. Barrow, \emph{{Cosmological limits on slightly skew stresses}},
  \href{https://doi.org/10.1103/PhysRevD.55.7451}{\emph{Phys. Rev.} {\bfseries
  D55} (1997) 7451--7460},
  [\href{https://arxiv.org/abs/gr-qc/9701038}{{\ttfamily gr-qc/9701038}}].

\bibitem{LeBlanc:1997}
V.~G. LeBlanc, \emph{{Asymptotic states of magnetic Bianchi I cosmologies}},
  \href{https://doi.org/10.1088/0264-9381/14/8/025}{\emph{Class. Quant. Grav.}
  {\bfseries 14} (1997) 2281--2301}.

\bibitem{Kanno:2008}
S.~Kanno, M.~Kimura, J.~Soda and S.~Yokoyama, \emph{{Anisotropic Inflation from
  Vector Impurity}},
  \href{https://doi.org/10.1088/1475-7516/2008/08/034}{\emph{JCAP} {\bfseries
  0808} (2008) 034}, [\href{https://arxiv.org/abs/0806.2422}{{\ttfamily
  0806.2422}}].

\bibitem{Watanabe:2009}
M.-a. Watanabe, S.~Kanno and J.~Soda, \emph{{Inflationary Universe with
  Anisotropic Hair}},
  \href{https://doi.org/10.1103/PhysRevLett.102.191302}{\emph{Phys. Rev. Lett.}
  {\bfseries 102} (2009) 191302},
  [\href{https://arxiv.org/abs/0902.2833}{{\ttfamily 0902.2833}}].

\bibitem{Calogero:2009}
S.~Calogero and J.~M. Heinzle, \emph{{Bianchi Cosmologies with Anisotropic
  Matter: Locally Rotationally Symmetric Models}},
  \href{https://doi.org/10.1016/j.physd.2010.11.015}{\emph{Physica} {\bfseries
  D240} (2011) 636--669}, [\href{https://arxiv.org/abs/0911.0667}{{\ttfamily
  0911.0667}}].

\bibitem{Yamamoto:2011}
K.~Yamamoto, \emph{{Bianchi Class B Spacetimes with Electromagnetic Fields}},
  \href{https://doi.org/10.1103/PhysRevD.85.043510}{\emph{Phys. Rev.}
  {\bfseries D85} (2012) 043510},
  [\href{https://arxiv.org/abs/1108.5983}{{\ttfamily 1108.5983}}].

\bibitem{int:Thorsrud11}
S.~Hervik, D.~F. Mota and M.~Thorsrud, \emph{{Inflation with stable anisotropic
  hair: Is it cosmologically viable?}},
  \href{https://doi.org/10.1007/JHEP11(2011)146}{\emph{JHEP} {\bfseries 11}
  (2011) 146}, [\href{https://arxiv.org/abs/1109.3456}{{\ttfamily 1109.3456}}].

\bibitem{int:Thorsrud12}
M.~Thorsrud, D.~F. Mota and S.~Hervik, \emph{{Cosmology of a Scalar Field
  Coupled to Matter and an Isotropy-Violating Maxwell Field}},
  \href{https://doi.org/10.1007/JHEP10(2012)066}{\emph{JHEP} {\bfseries 10}
  (2012) 066}, [\href{https://arxiv.org/abs/1205.6261}{{\ttfamily 1205.6261}}].

\bibitem{Ben18}
B.~D. Normann, S.~Hervik, A.~Ricciardone and M.~Thorsrud, \emph{{Bianchi
  cosmologies with $p$-form gauge fields}},
  \href{https://doi.org/10.1088/1361-6382/aab3a7}{\emph{Class. Quant. Grav.}
  {\bfseries 35} (2018) 095004},
  [\href{https://arxiv.org/abs/1712.08752}{{\ttfamily 1712.08752}}].

\bibitem{Normann:2019_1}
B.~D. Normann and S.~Hervik, \emph{{Approaching Wonderland}},
  \href{https://arxiv.org/abs/1909.11962}{{\ttfamily 1909.11962}}.

\bibitem{Normann:2019_2}
B.~D. Normann and S.~Hervik, \emph{{Collins in Wonderland}},
  \href{https://arxiv.org/abs/1910.12083}{{\ttfamily 1910.12083}}.

\bibitem{Almeida:2019iqp}
J.~P. Beltrán~Almeida, A.~Guarnizo, R.~Kase, S.~Tsujikawa and C.~A.
  Valenzuela-Toledo, \emph{{Anisotropic $2$-form dark energy}},
  \href{https://doi.org/10.1016/j.physletb.2019.05.008}{\emph{Phys. Lett.}
  {\bfseries B793} (2019) 396--404},
  [\href{https://arxiv.org/abs/1902.05846}{{\ttfamily 1902.05846}}].

\bibitem{Barrow:2003fc}
J.~D. Barrow and S.~Hervik, \emph{{The Future of tilted Bianchi universes}},
  \href{https://doi.org/10.1088/0264-9381/20/13/329}{\emph{Class. Quant. Grav.}
  {\bfseries 20} (2003) 2841--2854},
  [\href{https://arxiv.org/abs/gr-qc/0304050}{{\ttfamily gr-qc/0304050}}].

\bibitem{Coley-Hervik:2005}
A.~Coley and S.~Hervik, \emph{{A Dynamical systems approach to the tilted
  Bianchi models of solvable type}},
  \href{https://doi.org/10.1088/0264-9381/22/3/009}{\emph{Class. Quant. Grav.}
  {\bfseries 22} (2005) 579--606},
  [\href{https://arxiv.org/abs/gr-qc/0409100}{{\ttfamily gr-qc/0409100}}].

\bibitem{Sundell:2015gra}
P.~Sundell and T.~Koivisto, \emph{{Anisotropic cosmology and inflation from a
  tilted Bianchi IX model}},
  \href{https://doi.org/10.1103/PhysRevD.92.123529}{\emph{Phys. Rev.}
  {\bfseries D92} (2015) 123529},
  [\href{https://arxiv.org/abs/1506.04715}{{\ttfamily 1506.04715}}].

\bibitem{barrow1984helium}
J.~D. Barrow, \emph{Helium formation in cosmologies with anisotropic
  curvature}, {\emph{Monthly Notices of the Royal Astronomical Society}
  {\bfseries 211} (1984) 221--227}.

\bibitem{Ehlers:1966ad}
J.~Ehlers, P.~Geren and R.~K. Sachs, \emph{{Isotropic solutions of the
  Einstein-Liouville equations}},
  \href{https://doi.org/10.1063/1.1664720}{\emph{J. Math. Phys.} {\bfseries 9}
  (1968) 1344--1349}.

\bibitem{Clarkson:1999yj}
C.~A. Clarkson and R.~Barrett, \emph{{Does the isotropy of the CMB imply a
  homogeneous universe? Some generalized EGS theorems}},
  \href{https://doi.org/10.1088/0264-9381/16/12/302}{\emph{Class. Quant. Grav.}
  {\bfseries 16} (1999) 3781--3794},
  [\href{https://arxiv.org/abs/gr-qc/9906097}{{\ttfamily gr-qc/9906097}}].

\bibitem{int:Hasse88}
W.~Hasse and V.~Perlick, \emph{{Geometrical and kinematical characterization of
  parallax-free world models}},
  \href{https://doi.org/10.1063/1.527863}{\emph{J. Math. Phys.} {\bfseries 29}
  (1988) 2064--2068}.

\bibitem{pert:Zlosnik11}
T.~G. Zlosnik, \emph{{Cosmological Perturbation Theory With Background
  Anisotropic Curvature}},  \href{https://arxiv.org/abs/1107.0389}{{\ttfamily
  1107.0389}}.

\bibitem{pert:Pereira12}
T.~S. Pereira, S.~Carneiro and G.~A.~M. Marug\'an, \emph{{Inflationary
  Perturbations in Anisotropic, Shear-Free Universes}},
  \href{https://doi.org/10.1088/1475-7516/2012/05/040}{\emph{JCAP} {\bfseries
  1205} (2012) 040}, [\href{https://arxiv.org/abs/1203.2072}{{\ttfamily
  1203.2072}}].

\bibitem{Pereira:2015pxa}
T.~S. Pereira, G.~A.~M. Marugán and S.~Carneiro, \emph{{Cosmological
  Signatures of Anisotropic Spatial Curvature}},
  \href{https://doi.org/10.1088/1475-7516/2015/07/029}{\emph{JCAP} {\bfseries
  1507} (2015) 029}, [\href{https://arxiv.org/abs/1505.00794}{{\ttfamily
  1505.00794}}].

\bibitem{pert:Pereira17}
F.~O. Franco and T.~S. Pereira, \emph{{Tensor Perturbations in Anisotropically
  Curved Cosmologies}},
  \href{https://doi.org/10.1088/1475-7516/2017/11/022}{\emph{JCAP} {\bfseries
  1711} (2017) 022}, [\href{https://arxiv.org/abs/1709.00007}{{\ttfamily
  1709.00007}}].

\bibitem{main:Mimoso93}
J.~P. Mimoso and P.~Crawford, \emph{{Shear - free anisotropic cosmological
  models}}, \href{https://doi.org/10.1088/0264-9381/10/2/013}{\emph{Class.
  Quant. Grav.} {\bfseries 10} (1993) 315--326}.

\bibitem{main:Coley94}
A.~A. Coley and D.~J. McManus, \emph{{On space-times admitting shear - free,
  irrotational, geodesic timelike congruences}},
  \href{https://doi.org/10.1088/0264-9381/11/5/013}{\emph{Class. Quant. Grav.}
  {\bfseries 11} (1994) 1261--1282},
  [\href{https://arxiv.org/abs/gr-qc/9405034}{{\ttfamily gr-qc/9405034}}].

\bibitem{main:Coley94nr2}
D.~J. McManus and A.~A. Coley, \emph{{Shear - free irrotational, geodesic,
  anisotropic fluid cosmologies}},
  \href{https://doi.org/10.1088/0264-9381/11/8/011}{\emph{Class. Quant. Grav.}
  {\bfseries 11} (1994) 2045--2058},
  [\href{https://arxiv.org/abs/gr-qc/9405035}{{\ttfamily gr-qc/9405035}}].

\bibitem{shear-free:Abebe16}
A.~Abebe, D.~Momeni and R.~Myrzakulov, \emph{{Shear-free anisotropic
  cosmological models in f(R) gravity}},
  \href{https://doi.org/10.1007/s10714-016-2046-1}{\emph{Gen. Rel. Grav.}
  {\bfseries 48} (2016) 49},
  [\href{https://arxiv.org/abs/1507.03265}{{\ttfamily 1507.03265}}].

\bibitem{main:Koivisto11}
T.~S. Koivisto, D.~F. Mota, M.~Quartin and T.~G. Zlosnik, \emph{{On the
  Possibility of Anisotropic Curvature in Cosmology}},
  \href{https://doi.org/10.1103/PhysRevD.83.023509}{\emph{Phys.Rev.} {\bfseries
  D83} (2011) 023509}, [\href{https://arxiv.org/abs/1006.3321}{{\ttfamily
  1006.3321}}].

\bibitem{Menezes:2012kc}
R.~S. Menezes, Jr, C.~Pigozzo and S.~Carneiro, \emph{{Distance-Redshift
  Relations in an Anisotropic Cosmological Model}},
  \href{https://doi.org/10.1088/1475-7516/2013/03/033}{\emph{JCAP} {\bfseries
  1303} (2013) 033}, [\href{https://arxiv.org/abs/1210.2909}{{\ttfamily
  1210.2909}}].

\bibitem{main:Collins82}
C.~B. Collins and J.~Wainwright, \emph{{Role of shear in general relativistic
  cosmological and stellar models}},
  \href{https://doi.org/10.1103/PhysRevD.27.1209}{\emph{Phys. Rev.} {\bfseries
  D27} (1983) 1209--1218}.

\bibitem{main:Carneiro01}
S.~Carneiro and G.~A.~M. Marug\'an, \emph{{Anisotropic cosmologies containing
  isotropic background radiation}},
  \href{https://doi.org/10.1103/PhysRevD.64.083502}{\emph{Phys. Rev.}
  {\bfseries D64} (2001) 083502},
  [\href{https://arxiv.org/abs/gr-qc/0109039}{{\ttfamily gr-qc/0109039}}].

\bibitem{Mik18}
M.~Thorsrud, \emph{{Balancing Anisotropic Curvature with Gauge Fields in a
  Class of Shear-Free Cosmological Models}},
  \href{https://doi.org/10.1088/1361-6382/aab65a}{\emph{Class. Quant. Grav.}
  {\bfseries 35} (2018) 095011},
  [\href{https://arxiv.org/abs/1712.02778}{{\ttfamily 1712.02778}}].

\bibitem{Pailas:2019}
T.~Pailas and T.~Christodoulakis, \emph{{Dynamically equivalent $\Lambda$CDM
  equations with underlying Bianchi Type geometry}},
  \href{https://doi.org/10.1088/1475-7516/2019/07/029}{\emph{JCAP} {\bfseries
  1907} (2019) 029}, [\href{https://arxiv.org/abs/1903.07473}{{\ttfamily
  1903.07473}}].

\bibitem{Yamamoto:2012}
K.~Yamamoto, M.-a. Watanabe and J.~Soda, \emph{{Inflation with
  Multi-Vector-Hair: The Fate of Anisotropy}},
  \href{https://doi.org/10.1088/0264-9381/29/14/145008}{\emph{Class. Quant.
  Grav.} {\bfseries 29} (2012) 145008},
  [\href{https://arxiv.org/abs/1201.5309}{{\ttfamily 1201.5309}}].

\bibitem{Adamek:2011}
J.~Adamek, R.~Durrer, E.~Fenu and M.~Vonlanthen, \emph{{A large scale coherent
  magnetic field: interactions with free streaming particles and limits from
  the CMB}}, \href{https://doi.org/10.1088/1475-7516/2011/06/017}{\emph{JCAP}
  {\bfseries 1106} (2011) 017},
  [\href{https://arxiv.org/abs/1102.5235}{{\ttfamily 1102.5235}}].

\bibitem{bok:MacCallum79}
M.~A.~H. MacCallumm, \emph{{Anisotropic and inhomogeneous relativistic
  cosmologies}}.
\newblock In Hawking, S.W. and Israel, W., editors, General relativity: an
  Einstein centenary survey. Cambridge University Press, 1979.

\bibitem{bok:Hervik}
{\O}.~Gr{\o}n and S.~Hervik, \emph{{Einstein's general theory of relativity:
  With modern applications in cosmology}}.
\newblock Springer, 2007.

\bibitem{bok:Ellis98}
G.~F.~R. Ellis and H.~van Elst, \emph{{Cosmological models: Cargese lectures
  1998}}, {\emph{NATO Adv.Study Inst.Ser.C.Math.Phys.Sci.} {\bfseries 541}
  (1999) 1--116}, [\href{https://arxiv.org/abs/gr-qc/9812046}{{\ttfamily
  gr-qc/9812046}}].

\bibitem{bok:KantowskiSachs66}
R.~Kantowski and R.~K. Sachs, \emph{{Some spatially homogeneous anisotropic
  relativistic cosmological models}},
  \href{https://doi.org/10.1063/1.1704952}{\emph{J. Math. Phys.} {\bfseries 7}
  (1966) 443}.

\bibitem{bok:Kantowski66}
R.~Kantowski, \emph{{Some relativistic cosmological models}}.
\newblock University of Texas PhD thesis, 1966.

\bibitem{Mik19}
M.~Thorsrud, \emph{{Bianchi models with a free massless scalar field: invariant
  sets and higher symmetries}},
  \href{https://doi.org/10.1088/1361-6382/ab45b3}{\emph{Class. Quant. Grav.}
  {\bfseries 36} (2019) 235014},
  [\href{https://arxiv.org/abs/1905.11456}{{\ttfamily 1905.11456}}].

\bibitem{Pereira:2016tmu}
T.~S. Pereira and D.~T. Pabon, \emph{{Extending the $\Lambda$CDM model through
  shear-free anisotropies}},
  \href{https://doi.org/10.1142/S0217732316400095}{\emph{Mod. Phys. Lett.}
  {\bfseries A31} (2016) 1640009},
  [\href{https://arxiv.org/abs/1603.04291}{{\ttfamily 1603.04291}}].

\bibitem{Jaffe:2005}
T.~R. Jaffe, S.~Hervik, A.~J. Banday and K.~M. Gorski, \emph{{On the viability
  of Bianchi type viih models with dark energy}},
  \href{https://doi.org/10.1086/503893}{\emph{Astrophys. J.} {\bfseries 644}
  (2006) 701--708}, [\href{https://arxiv.org/abs/astro-ph/0512433}{{\ttfamily
  astro-ph/0512433}}].

\end{thebibliography}\endgroup

\end{document}